\let\oldvec\vec
\documentclass[runningheads,a4paper]{llncs}
\let\vec\oldvec

\usepackage{microtype}
\usepackage{color}
\usepackage{subfigure}
\usepackage{amssymb}
\usepackage{graphicx}
\usepackage[linesnumbered,ruled]{algorithm2e}
\usepackage{dsfont}
\usepackage{amsfonts}
\usepackage{extarrows}
\usepackage{multirow}
\usepackage{diagbox}

\newcommand{\transto}[1]{\xrightarrow{\,{#1}\,}}

\def\>{\ensuremath{\rangle}}
\def\<{\ensuremath{\langle}}
\def\ra{\ensuremath{\rightarrow}}

\newcommand {\supp } {{\rm supp}}
\newcommand {\conf}{{\rm conf}}

\newcommand {\E } {{\mathcal{E}}}

\newcommand{\hs}{\mathcal{H}}

\newcommand {\tr} {{\mathrm{tr}}}

\newcommand {\trans} {{tr}}

\title{Quantum Privacy-Preserving Data Analytics}

\author{Shenggang Ying\inst{1}, Mingsheng Ying\inst{1,2,3}, Yuan Feng\inst{1}}
\institute{University of Technology Sydney, Australia
\and Institute of Software, Chinese Academy of Sciences, China
\and Tsinghua University, China
}

\begin{document}

\maketitle

\begin{abstract}
Data analytics (such as association rule mining and decision tree mining) can discover useful statistical knowledge from a big data set. But protecting the privacy of the data provider and the data user in the process of analytics is a serious issue. Usually, the privacy of both parties cannot be fully protected simultaneously by a classical algorithm. In this paper, we present a quantum protocol for data mining that can much better protect privacy than the known classical algorithms: (1) if both the data provider and the data user are honest, the data user can know nothing about the database except the statistical results, and the data provider can get nearly no information about the results mined by the data user;  (2)   if the data user is dishonest and tries to disclose private information of the other, she/he will be detected with a high probability; (3) if the data provider tries to disclose the privacy of the data user, she/he cannot get any useful information since the data user hides his privacy among noises.
\end{abstract}

\section{Introduction}

\ \ \ \ \textbf{Privacy-preserving data analytics}: Data analytics has become an indispensable technology in the big data era. Mining statistical knowledge from a big data set is one of the most important tasks of data analytics. A typical example is association rule mining, which was introduced to find useful links from a large set of transactions in a supermarket \cite{AgrawalIS1993}. An association rule is a probabilistic implication $A\Rightarrow B$, which means event $A$ implies event $B$ with a high probability. Another example is to mine decision trees \cite{Quinlan1986}, which are a core model of classification problems.

Data analytics has numerous applications in the areas like market basket problem, scientific data analysis, web mining, just name a few \cite{AgrawalIS1993,KotsiantisK2006}.
In practical applications, one major issue arises: how to protect the privacy of each individual in a database while mining the statistical knowledge? For instance, the privacy of each patient should not be leaked during mining an association rule or a decision tree for medical diagnosis from a database of patients.
To address this issue, various algorithms for privacy-preserving data mining has been developed in the last twenty years \cite{EstivillB1999,EvfimievskiSAG2004,KarguptaDW2003,RizviH2002}. In these algorithms, however, the privacy of the data provider and the data user cannot be protected simultaneously.

{\vskip 4pt}

\textbf{Quantum computing and crytopgraphy}: Since 1990's, various quantum algorithms have been discovered and proved to be much faster than the known classical algorithms for the same tasks. For example, Grover's quantum search algorithm \cite{Grover1996} can find the target element in a database in $O(\sqrt{N})$ oracle calls. Quantum counting algorithm \cite{BrassardHT1998} has a quadratic speed-up over classical algorithms as well. More recently, several quantum machine learning algorithms \cite{LloydMR2013ML,LloydMR2014PA} have been presented based on quantum random access memory \cite{GiovannettiLM2008qram}, and they can achieve an exponential speed-up over classical algorithms.

Several quantum protocols that can better protect privacy have also been found; for example, the famous quantum key distribution protocol BB84 \cite{BB84}, quantum private queries \cite{GiovannettiLM2008PQ}, and revocable quantum timed-release encryption \cite{Unruh15}.

{\vskip 4pt}

\textbf{Contribution of this paper}: In this paper, we present a quantum protocol for mining statistical knowledge in a database, such as association rules and decision trees. This protocol can protect the privacy of both the data provider and data user provided they are honest. Furthermore, {  the privacy of both data provider and data user is protected: no useful private information will be disclosed}.

The basic idea of our protocol can be described as follows. {  The basic idea for the data provider is to employ tests to detect attacks.} Without any influence on the computational results, the data provider randomly detects attacks from the data user. Meanwhile, {  the basic idea for the data user is to hide his privacy among noises.} The data user introduces noises into her/his private query functions, and these noises in different steps cancel each other if the data provider follows the protocol strictly. The novelty is that the privacy of both parties is preserved by techniques in quantum computing and cryptography, but can hardly be achieved by classical methods.

{\vskip 4pt}

\textbf{Structure of the paper}: For convenience of the reader, preliminaries and notations are introduced in Section \ref{Sec:Preliminaries}. We present our quantum protocol in three steps: we first explain the design idea in Section \ref{Sec:BasicIdea}, an outline of the protocol is then shown in Section \ref{main protocol}, and we examine some details in the execution of the protocol in Section \ref{Sec:ProtocolExecution}.
The correctness of the protocol is proven in Section \ref{Sec:Correctness}. The privacy analysis is given in Section \ref{Sec:PrivacyAlice} for the data provider and in Section \ref{Sec:PrivacyBob} for the data user. The complexity analysis is given in Section \ref{Sec:Complexity}. Some further discussions are presented in Section \ref{Sec:Discussions}. All the proofs of lemmas and theorems are given in the Appendix.

\section{Preliminaries}\label{Sec:Preliminaries}

\subsection{Association Rule Mining}

As pointed out above, we are going to develop a quantum algorithm for association rule mining. {  (The application for decision tree learning is presented in Section \ref{sec:DecisionTree}.)} For convenience of the reader, in this subsection, we briefly review association rule mining; for more details, see \cite{AgrawalIS1993}. Let $S=\{1,2,\cdots, k\}$ be a set of \textbf{items}, where each index $i\in S$ stands for an item; for instance, $1$ may stand for ``Apple'', $2$ for ``Orange''. A \textbf{transaction} or \textbf{itemset} $\trans$ is a set of items, i.e., $\trans\subseteq S$. Moreover, an $m$-\textbf{itemset} is an itemset which has exactly $m$ items. For example, a transaction can be the items a customer buys in one purchase. In order to store a transaction into a computer or a database, a transaction $\trans$ is represented by a string $\pi=\pi_1 \pi_2\cdots \pi_k\in\{0,1\}^k$:
\begin{equation}\label{eq:TransactionString}
    \pi_\iota = \begin{cases}
        1, & \iota \in \trans,\\
        0, & \iota \not\in \trans.
    \end{cases}
\end{equation}
In this paper, we always use strings in $\{0,1\}^k$ to represent transactions or itemsets based on Eq. \eqref{eq:TransactionString}.
Then we can talk about the \textbf{inclusion relation} $\pi\subseteq \tau$ for two strings $\pi,\tau\in\{0,1\}^k$ since they refer to two sets (transactions or itemssets): $\pi\subseteq \tau \Leftrightarrow \pi_\iota\leq \tau_\iota\ {\rm for\ every}\ \iota\in S.$
Moreover, we can define other set-theoretic operations and relations of $\pi$ and $\tau$.

A \textbf{database} $D$ of transactions is a set $D=\<d_0,d_1,\cdots,d_{N-1}\>$, where $d_j\in\{0,1\}^k$ is a transaction and $N$ is the size of $D$, i.e., the number of transactions in $D$. The \textbf{support} $\supp(d)$ of an itemset $d$ is defined to be its frequency in database $D$:
\begin{equation}\label{eq:fD}
    \supp(d) = f^{(d)}(D) =\frac 1N   \sum_{j=0}^{N-1} f^{(d)}(d_j),
\end{equation}
where
\begin{equation}\label{eq:fD1}
    f^{(d)}(d_j) = \begin{cases}
        1, & d\subseteq d_j,\\
        0, & d\not\subseteq d_j.
    \end{cases}
\end{equation}
The superscript $(d)$ of $f^{(d)}(\cdot)$ may be omitted if the itemset $d$ is clear from the context or not explicitly specified.

A \textbf{rule} is a probabilistic implication between two disjoint itemsets in $D$. The support and confidence of a rule $\pi\Rightarrow \tau$ are defined by $$\supp(\pi\Rightarrow \tau)=\supp(\pi\cup\tau),$$  \begin{equation*}
    \conf(\pi\Rightarrow \tau) =\frac{\supp(\pi\cup\tau)}{\supp(\pi)},
\end{equation*} respectively.
Roughly speaking, the support of a rule indicates its importance (frequency) in a database. The confidence of a rule means its correctness probability; more precisely, $\conf(\pi\Rightarrow \tau)$ is the probability that if $\pi$ appears in $d$ ($\pi\subseteq d$), then $\tau$ will as well appear in $d$ ($\tau\subseteq d$).

The task of association rule mining is to find all rules with high support and confidence, namely \textbf{association rule}.
\begin{definition}[Association Rule]
    An association rule $\pi\Rightarrow \tau$ is a relation between two itemsets $\pi,\tau\subseteq S$, which satisfies the following conditions:
    \begin{itemize}
      \item $\pi\cap\tau = \emptyset$,
      \item $\supp(\pi\Rightarrow\tau)>s_{\min}$,
      \item $\conf(\pi\Rightarrow\tau)>c_{\min}$,
    \end{itemize}
    where $s_{\min}$ (resp. $c_{\min}$) is the preset (constant) threshold for supports (resp. confidences).
\end{definition}
The algorithms for association rule mining in the literature \cite{AgrawalIS1993} are usually divided into the following two steps:
\begin{enumerate}
  \item Find all \textbf{frequent itemsets} $\xi$, which are itemsets with high support/frequency, i.e., $\supp(\xi)>s_{\min}$.
  \item For each frequent itemsset $\xi$, find all association rules $\pi\Rightarrow \tau$ with $\pi\cup\tau=\xi$.
\end{enumerate}
Moreover, only the first step is crucial because once it is done, the second becomes very easy by the definition of confidence. 
Note $\xi\subseteq d$ implies $\tau\subseteq d$
for any $\tau\subseteq \xi$. So, the supports of itemsets are non-increasing: $\tau\subseteq\xi$ implies $\supp(\tau)\geq\supp(\xi)$. Thus, each frequent $m$-itemset is a superset of some frequent $(m-1)$-itemset, which leads to the level-wise algorithm \cite{AgrawalIS1993}:
\begin{itemize}
  \item Initialization. Let $F_1 = \{\{i\}:  i = 1, 2,\cdots,k\}$ be the set of all 1-itemsets. Set $F_l=\emptyset$ for all $l>1$, and $G_{j}=\emptyset$ for all $j\geq 1$.
  \item {  Induction on $l$ starting from $l=1$. If $F_l=\emptyset$, output all itemsets in $G_{j}$ for all $j$, and terminate the algorithm. Otherwise, do the following steps for every $\tau\in F_l$:}
    \begin{enumerate}
      \item Compute $\supp(\tau)$.
      \item If $\supp(\tau)>s_{\min}$, add $\tau$ into $G_l$ and all supersets of $\tau$ with $l+1$ items into $F_{l+1}$.
    \end{enumerate}
\end{itemize}

One can see that in the above algorithm, the key step is to compute the support of a given itemset. {  In particular, only this step may cause disclosure of private data.} So, in this paper, we will focus on it. Note that computing the support of a given itemset can be done by quantum counting algorithm \cite{BrassardHT1998}. Our work is to develop a privacy-preserving extension of quantum counting for this task on a centralized database.

\subsection{Quantum Database}

To employ quantum algorithms for mining classical database, we first recall  from \cite{GiovannettiLM2008qram} how a (classical) database can be stored in a quantum computer.
\begin{definition}[Quantum database]\label{def:qdb} Let $D=\<d_0,d_1,...,d_{N-1}\>$ be a database where $d_j \in \{0,1\}^k$ for all $j$. For convenience, we always assume that $N=2^n$. Then the quantum database corresponding to $D$ is a unitary transformation $O_D$ on $n+k$ qubits defined as follows:
        $$O_D|x\>|\tau\> = |x\>|\tau\oplus d_x\>$$
    for every $x\in \{0,1\}^n$ and $\tau\in\{0,1\}^k$.
    Here we identify $x$ with the integer it represents (by binary representation).
    \end{definition}

In the above definition, $x$ is used to denote the address of a data cell, and $d_x$ is the content stored in data cell $x$. The quantum database $O_D$ can be seen as a quantum oracle; for instance, if a data user queries it with a basis state $|i\rangle|0\rangle$, the oracle will return $|i\rangle|d_i\rangle$. This is equivalent to querying the classical database $D$ with address $i$. More interestingly, a data user can also query $O_D$ with a superposition $\sum_i \alpha_i|i\>|0\>$ of addresses, and it will return
\begin{equation}
    O_D\left(\sum_i \alpha_i|i\>|0\>\right) = \sum_i \alpha_i|i\>|d_i\>,
\end{equation}
a (superposed) state which in principle contains \emph{all} transactions of the database. Note also that, however, an attempt to read out a particular transaction (by performing a quantum measurement) will cause collapse of the state into one where all information of other transactions is completely destroyed.

It is worth mentioning that in our protocol the quantum database $O_D$ will be permuted to $U_{D}(y)$ (see Eq. \eqref{eq:UDx}) to improve the data provider's privacy.

\subsection{Privacy-preserving Data Analytics}
Now let us consider the following problem:
\begin{problem}[Privacy-preserving counting]\label{Prob:count}
 Suppose Alice holds a database $D=\<d_0,d_1,...,d_{N-1}\>$ where $d_j \in \{0,1\}^k$ for all $j$. For a function $f : \{0,1\}^k \rightarrow \{0,1\}$, Bob wants to compute
 \begin{equation}\label{eq:fD2}
 f(D) =   \frac{1}{N}\sum_{j=0}^{N-1} f(d_j),
 \end{equation}
and after the computation,
        \begin{enumerate}
            \item \emph{Privacy preserving for Alice}: for each $j$, Bob will not know $d_j$ (even approximatively);
            \item \emph{Privacy preserving for Bob}: Alice will not know the function $f$.
        \end{enumerate}
\end{problem}

Obviously, whenever $f$ is taken to be $f^{(d)}$ defined in Eq. (\ref{eq:fD1}), then $f(D)$ is the support of item set $d$.

The majority of classical algorithms in the literature aim at protecting Alice's privacy. The idea is to publish a distorted database $D'$ so that Bob can compute $f$ over it. Thus, the problem becomes how to modify $D$ to $D'$ with a high accuracy of statistical properties. The suggested solutions include: (1) modify each transaction independently; for instance, the occurrence of each item in a transaction is flipped randomly \cite{EvfimievskiSAG2004}; (2) replace some items by others without changing the number of items in a transaction \cite{RizviH2002}; (3) modify transactions within the entire database; for instance, swap elements between different transactions \cite{EstivillB1999}.

It is easy to see that the function in Eq.(\ref{eq:fD}) can be efficiently computed by quantum counting algorithm \cite{BrassardHT1998} with the corresponding quantum database $O_D$.
However, a simple application of quantum counting is unable to achieve the goal of privacy protection. {  It has to be extended to fit the new task.}

\section{Basic Ideas of the Protocol}\label{Sec:BasicIdea}

The overall aim of this paper is to develop a quantum protocol solving Problem~\ref{Prob:count}. In this section, we introduce the basic ideas employed in the design of our protocol, which is essentially the quantum counting algorithm \cite{BrassardHT1998} with new strategies for privacy preserving. Quantum counting is a combination of controlled Grover iterations modified from Grover search algorithm \cite{Grover1996} and quantum Fourier transform \cite{NielsenC2000}. As quantum Fourier transform is not applied to the original data set, privacy preserving in Grover search is the core of our protocol.

To explain the ideas more clearly, we elaborate in the following a quantum algorithm for \emph{privacy-preserving search}, a problem which can be regarded as a special case of
Problem~\ref{Prob:count} where $f(D)$ returns an index (if there is any) $j$ such that $f(d_j)=1$.


\subsection{Two-Party Grover Search}\label{sec:tpgrover}
 For privacy preserving, we adapt Grover search algorithm~\cite{Grover1996} to the two-party scenario. To simplify the presentation, we omit
the detail of communication between Alice and Bob, and assume implicitly that when Alice (resp. Bob) performs a quantum operation, the corresponding quantum system has been sent to her (resp. him) by Bob (resp. Alice) or prepared by herself (resp. himself). The algorithm goes as follows:

\begin{itemize}
  \item Bob prepares the initial state $|+\>_{q_a}^{\otimes n}|0\>_{q_a}^{\otimes k}=\frac{1}{\sqrt{N}}\sum_{j=0}^{N-1} |j\>_{q_a}|\vec{0}\>_{q_d}$
  where $q_a$ denotes the ($n$-qubit) address system while $q_d$ the ($k$-qubit)  data system, $N=2^n$, and $\vec{0}$ is the $k$-length binary representation of $0$.
  \item Repeat the following steps for $T=\lceil \frac{\pi}{4}\sqrt{N} \rceil$ times:
    \begin{enumerate}
      \item Alice applies the database $O_D$ on systems $q_a$ and $q_d$. Let $|\phi_0\>=\sum_j \alpha_j|j\>|\vec{0}\>$ be the initial state of the current iteration. Then now it becomes
        \begin{equation*}
            |\phi_1\>=O_D|\phi_0\> = \sum_j \alpha_j|j\>|d_j\>.
        \end{equation*}
      \item Bob applies $U_f$, obtaining
        \begin{equation*}
            |\phi_2\> = U_f|\phi_1\> = \sum_j (-1)^{f(d_j)}\alpha_j|j\>|d_j\>,
        \end{equation*}
        where $U_f$ is the oracle defined by:
        \begin{equation}\label{eq:Uf}
            U_f:|j\>|\tau\> \mapsto (-1)^{f(\tau)}|j\>|\tau\>.
        \end{equation}
      \item Alice applies $O_D$ again to disentangle the systems $q_a$ and $q_d$, reaching
        \begin{equation*}
            |\phi_3\> = O_D|\phi_2\> = \sum_j (-1)^{f(d_j)}\alpha_j|j\>|\vec{0}\>.
        \end{equation*}
      \item Bob performs $G$ on system $q_a$ only to update the amplitude, obtaining
        \begin{equation*}
            |\phi'_0\> = G\otimes I_{q_d}|\phi_3\> = \sum_j \alpha'_j|j\>|\vec{0}\>,
        \end{equation*}
        where $G = I - 2|+\>^{\otimes n}\<+|^{\otimes n}$. The state $|\phi'_0\>$ will be the initial state of the next iteration.
    \end{enumerate}
  \item Another iteration of the above loop is executed with $U_f$ in Eq.(\ref{eq:Uf}) replaced by $U'_{f}$, and it is applied on $q_a$, $q_d$ and an auxiliary qubit system $q_{g}$ which has been set to $|0\>$, where
    \begin{equation}\label{eq:Uf2}
        U'_f:|j\>_{q_a}|\tau\>_{q_d}|x\>_{q_g} \mapsto|j\>_{q_a}|\tau\>_{q_d}|x\oplus f(\tau)\>_{q_g}.
    \end{equation}
 \item Bob measures $q_a$ and $q_g$, and reports the measurement outcome.

\end{itemize}

By a similar argument as that given in \cite{Grover1996}, we can show that the above algorithm returns an index $j$ with $f(d_j)=1$ with a high probability, provided that both Alice and Bob follow the protocol honestly.

\subsection{Possible Attacks}

\subsubsection{Bob's attack:} An obvious Bob's attack for the above algorithm is to send state $|j\>|\vec{0}\>$ for a chosen $j$ to Alice before Step 1 in the loop. Then an honest Alice will send back to him $|j\>|d_j\>$, from which he is able to successfully disclose $d_j$. Note that in one run of the algorithm, Bob may cheat $2T$ times. Thus if it is called $M$ times as a procedure in, say, association-rule mining protocol, Bob will obtain complete information of $2TM$ transactions of the database.

\subsubsection{Alice's attack:} Similarly, Alice can cheat by sending some chosen states to Bob to retrieve information about the query function $f$. To see this, note that in both association rule and decision tree mining, $f(\vec{1})=1$ for all legal $f$, where $\vec{1}$ is the $k$-length binary representation of $2^k-1$. Now suppose Alice would like to know the value of $f(\tau)$ for some $\tau\in \{0,1\}^k$. Then she chooses to send Bob the
state $\frac{1}{\sqrt{2}}|0\>^{\otimes (n-1)}(|0\>|\tau\>+|1\>|\vec{1}\>)$ before Step 2 of the loop. Note that
$$    \frac{1}{\sqrt{2}}|0\>^{\otimes (n-1)}(|0\>|\tau\>+|1\>|\vec{1}\>) \transto{U_f} \frac{1}{\sqrt{2}}|0\>^{\otimes (n-1)}[(-1)^{f(\tau)}|0\>|\tau\>-|1\>|\vec{1}\>].$$
Now Alice can obtain $f(\tau)$ by performing a quantum measurement on systems $q_a$ and $q_d$, since the states $|0\>|\tau\>-|1\>|\vec{1}\>$ and $|0\>|\tau\>+|1\>|\vec{1}\>$ are orthogonal.

\subsection{Privacy Preserving in Quantum Search}\label{ideas-0}

This subsection is devoted to an intuitive explanation of the techniques we are going to employ in the protocol presented in Section~\ref{main protocol} to protect the privacy of both Alice and Bob.

\subsubsection{Alice's strategy:} {  The idea for protecting Alice's privacy is to employ tests to detect Bob's attacks. Originally it is hard to distinguish Bob's attacks from honest actions, since Alice does not know Bob's function $f$. Fortunately, if Bob is honest, two same states will still be the same after Bob's action $U_f$. Suppose Alice randomly generates two different strings $\mu,\nu\in\{0,1\}^k$, and $f(\mu)=f(\nu)$ (resp. $f(\mu)\neq f(\nu)$). Alice sends Bob two same states $\frac{1}{\sqrt{2}}|0\>^{\otimes n-1}(|0\>|\mu\>+|1\>|\nu\>)$ to Bob. Then she will receive two same states $\frac{1}{\sqrt{2}}|0\>^{\otimes n-1}(|0\>|\mu\>+|1\>|\nu\>)$ (resp. $\frac{1}{\sqrt{2}}|0\>^{\otimes n-1}(|0\>|\mu\>-|1\>|\nu\>)$) from Bob. After disentangling the data qubits, Alice gets two copies of  $|0\>^{\otimes n-1}|+\>|\vec{0}\>$ (resp. $|0\>^{\otimes n-1}|-\>|\vec{0}\>$). Finally Alice performs measurements on the last address qubits with basis $\{|+\>,|-\>\}$, and gets outcomes $+,+$ (resp. $-,-$).

But if Bob is dishonest and performs measurements on data qubits to read information, each state that Alice receives will be $|0\>^{\otimes n-1}|0\>|\mu\>$ or $|0\>^{\otimes n-1}|1\>|\nu\>$ independently. Finally the measurement outcomes on the last address qubits will be $+,-$ or $-,+$ with probability 0.5. These outcomes can be distinguished from those of honest actions.}

\subsubsection{Bob's strategy:} The idea for protecting Bob's privacy is to add noises which cancel each other when Alice follows the protocol honestly. Recall that the net effect of a single iteration of the loop in the algorithm presented in Section~\ref{sec:tpgrover} is $\bar{G}O_DU_f O_D$ where $\bar{G} = G\otimes I_d$. Then in four consecutive iterations, for instance, if the four calls of oracle $U_f$ at Step 2 are replaced by $U$, $I_{a,d}$, $U$, $I_{a,d}$, respectively, where $U$ is any unitary operator with $U=U^\dag$ and $I_{a,d}$ is the identity operator on $q_a$ and $q_d$, then the effect of the four modified iterations becomes
\begin{equation}\label{law-0}\bar{G}O_DI_{a,d} O_D\bar{G}O_DU O_D\bar{G}O_DI_{a,d} O_D\bar{G}O_DU O_D = I_{a,d}.\end{equation}
More generally, if Bob needs to use $U_f$ for $T$ times, he can insert $T$ operators $I_{a,d}$ and $U_{f'}$ with different $f'\neq f$ between the $T$ occurrences of $U_f$. By Eq.(\ref{law-0}), we see that half of the information Alice gets is noise and not related to $f$, and thus she cannot recover $f$.

\section{Protocol}\label{main protocol}

\subsection{Main Protocol}
\begin{algorithm}
    \SetKwData{Left}{left}\SetKwData{This}{this}\SetKwData{Up}{up}
    \SetKwFunction{Union}{Union}\SetKwFunction{FindCompress}{FindCompress}
    \SetKwInOut{Input}{input}\SetKwInOut{Output}{Output}\SetKwInOut{Parameter}{Parameters}

    \Parameter{Number of address qubits $n$ and number of data qubits $k$ determined by the database $D$, and number of control qubits $t$ determined by Bob's strategy in Section \ref{Sec:BobStrategyInner}.}
    \Output{  two numbers $s_1,s_2\in [0,1]$, which are approximately $f(D)$.}
    \Begin{
    Alice generates $y\in\{0,1\}^n$ uniformly at random, and constructs a modified database $U_D(y)$ from $O_D$; see Eq. \eqref{eq:UDx}\;\label{line:B1}
    Bob prepares two identical states $|+\>_{q_{c1}}^{\otimes t}|+\>_{q_{a1}}^{\otimes n}|\vec{0}\>_{q_{d1}}$ and $|+\>_{q_{c2}}^{\otimes t}|+\>_{q_{a2}}^{\otimes n}|\vec{0}\>_{q_{d2}}$. Here
    $q_{ci}$, $q_{ai}$, and $q_{di}$ denote the control, address, and data qubits, respectively\;
    For $i=0, \ldots, T-1$, where $T=2^t$, do \ref{alg:ProtocolLoop}($i$)\;\label{line:B4}
    Alice generates $r\in(0,1)$ uniformly at random\;
    If $r\leq p$, Alice employs procedure \ref{alg:TestACDG} to test whether Bob is honest. If dishonesty is detected, she terminates the entire protocol; otherwise, she sends a message ``Repeat'' to Bob\;\label{line:BTerminating}
    Alice applies $U_D(y)$ on $q_{c1}$, $q_{a1}$, $q_{d1}$\;\label{line:B8}
    Bob performs $U'_f$ in Eq. \eqref{eq:Uf2} on $q_{a1}$,  $q_{d1}$, $q_{g1}$  where
    $q_{g1}$ is an auxiliary qubit which has been set to $|0\>$\;\label{line:B9}
    Alice and Bob repeat Step \ref{line:B8} to Step \ref{line:B9} for $q_{c2}$, $q_{a2}$, $q_{d2}$. The new blank ancilla qubit is denoted $q_{g2}$\;\label{line:B10}
    If $p<r\leq 2p$, Alice first sends a message ``Repeat'' to Bob, and then employs procedure \ref{alg:TestACDG} to test whether Bob is honest. If dishonesty is detected, Alice terminates the entire protocol\;\label{line:2p}
    Alice applies $U_D(y)$ on  $q_{a1}$, $q_{d1}$, and $q_{a2}$, $q_{d2}$ respectively\;
    Bob performs measurements on $q_{g1}$ and $q_{g2}$ to get outcome $g_1,g_2\in\{0,1\}$\;\label{line:BMg}
    Bob performs quantum Fourier transform, followed by measurements on $q_{c1}$ to get outcome $\theta\in\{0,1,\cdots, T-1\}$\;
    Bob computes $s_{i} = \sin^2(\theta\pi/T)$ if $g_i=1$, or $\cos^2(\theta\pi/T)$ if $g_i=0$, for $i=1,2$.\label{line:BFourier}
    }
    \caption{Main protocol for privacy-preserving quantum counting on a centralized database.}\label{alg:ProtocolInner}
\end{algorithm}

We are now ready to present our main protocol in Algorithm \ref{alg:ProtocolInner}, which computes $f(D)$ for a given function $f$ by applying procedure \ref{alg:ProtocolLoop}. In the protocol,
\begin{itemize}
  \item At Step \ref{line:B1}, $U_D(y)$ is defined as follows:
    \begin{equation}\label{eq:UDx}
        U_D(y) = (X^y\otimes I_d) O_D (X^y\otimes I_d).
    \end{equation}
    where $ X^y = X^{y_1}\otimes X^{y_2}\otimes\cdots\otimes X^{y_n}$, $X^0 = I$, and $X^1 = X$. Note that
    \begin{align*}
        U_D(y) |x\>|\tau\> = |x\>|\tau\oplus d_{x\oplus y}\>.
    \end{align*}
    This modification of $O_D$ permutes the transactions in the database to protect Alice's privacy (Compare it with $O_D$ in Definition~\ref{def:qdb}).

  \item  In the main protocol, Alice employs the test \ref{alg:TestACDG} for Bob's dishonesty with probability $2p$. Furthermore, the test is conducted either before or after Steps \ref{line:B8} to \ref{line:B10}, with equal probability. Consequently, both Step~\ref{line:BTerminating} and Step~\ref{line:2p} are executed with probability $p$. In this paper, we set $p= 0.05$.
  \item The message ``Repeat'' means that a test was or will be used, and Bob should prepare to repeat what he did for the last two copies of states.
\end{itemize}

\subsection{Grover Iteration}\label{Sec:GroverIteration}
In this subsection, we present the procedure \ref{alg:ProtocolLoop} called at Step \ref{line:B4} in Algorithm \ref{alg:ProtocolInner}, which is essentially a controlled Grover iteration. In this procedure,
\begin{itemize}
\item the functions $f_0, f_1, \ldots, f_{T-1}$ are generated by Bob using the strategy presented in Section \ref{Sec:BobStrategyInner}.
  \item for each $i$, $U_{f_i}$ is a unitary operator similarly defined as that in Eq.\eqref{eq:Uf}.
\end{itemize}

\begin{procedure}[t]
    \SetKwData{Left}{left}\SetKwData{This}{this}\SetKwData{Up}{up}
    \SetKwFunction{Union}{Union}\SetKwFunction{FindCompress}{FindCompress}
    \SetKwInOut{Input}{Input}\SetKwInOut{Output}{output}
    \Begin{
    Alice generates $r\in(0,1)$ uniformly at random\;
    If $r\leq p$, Alice employs procedure \ref{alg:TestACD1}($i$) to test whether Bob is honest. If dishonesty is detected, Alice terminates the entire protocol; otherwise, she sends a message ``Repeat'' to Bob\;\label{line:LTerminating1}
    Alice applies $U_D(y)$ and then Bob applies controlled $U_{f_i}$ on $q_{c1}$, $q_{a1}$, $q_{d1}$ with $q_{c1}$ being control qubits\;\label{line:L1}
    Alice and Bob repeat Step \ref{line:L1} for $q_{c2}$, $q_{a2}$, $q_{d2}$\;\label{line:L1r}
    If $p<r\leq 2p$, Alice first sends a message ``Repeat'' to Bob, and then employs procedure \ref{alg:TestACD1}($i$) to test whether Bob is honest. If dishonesty is detected, Alice terminates the entire protocol\;\label{line:LTerminating2}
    Alice applies $U_D(y)$ and then Bob applies controlled $G$ on $q_{c1}$, $q_{a1}$, $q_{d1}$ with $q_{c1}$ being control qubits, where $G$ is defined in Sec~\ref{sec:tpgrover}\;\label{line:L2}
    Alice and Bob repeat Step \ref{line:L2} for $q_{c2}$, $q_{a2}$, $q_{d2}$\;\label{line:L3}
    }
    \caption{GroverIteration($i$)}\label{alg:ProtocolLoop}
\end{procedure}

In the following, we discuss briefly some implementation issues for procedure \ref{alg:ProtocolLoop}.
\subsubsection{Controlled Operators.}
As said before, procedure \ref{alg:ProtocolLoop} is a controlled Grover iteration.
Two controlled unitary operators, controlled $U_{f_i}$ in Step~\ref{line:L1} and controlled $G$ in Step~\ref{line:L2}, need to be implemented.
{  Fortunately, both controlled $U_{f_i}$ and controlled $G$ can be implemented locally by Bob with $O(n+k)$ CNOT gates and Toffoli gates.}

\subsubsection{Control Qubits $q_c$.}
Note that $t=\log T$ qubits $q_c$ are used as the control bits for $U_{f_i}$ and $G$. Observe that for any unitary operator $U$,
\begin{equation}\label{eq:controlG}
    \frac{1}{\sqrt{T}}\sum_{c=0}^{T-1} |c\>U^{c} |\Psi_0\> = \frac{1}{\sqrt{T}}\sum_{c=0}^{T-1} |c\> \otimes (U^{c_0 2^{t-1}}U^{c_1 2^{t-2}}\cdots U^{c_{t-1} } |\Psi_0\>).
\end{equation}
We can use the first qubit of $q_c$ (corresponding to $c_0$) as the control for $2^{t-1}$ times {  (for loop $i=1,\cdots, 2^{t-1}$)}, and the second one for $2^{t-2}$ times {  (for loop $i=2^{t-1}+1,\cdots, 2^{t-1}+2^{t-2}$)}, and so on; see Table~\ref{Table:ControlQubit}. Note that one control qubit is enough {   to implement the controlled operators for each $i$}.
\begin{table}[t]
  \centering
  \begin{tabular}{l|c}
  \hline
  Control qubit & Loop $i$\\\hline
  None & 0 \\
  First qubit of $q_c$ & $1,\cdots, 2^{t-1}$\\
  Second qubit of $q_c$ & $2^{t-1}+1,\cdots, 2^{t-1}+2^{t-2}$\\
  \hspace{4em}\vdots & \vdots\\
  \hline
  \end{tabular}
  \caption{Control qubit for each loop $i$.}\label{Table:ControlQubit}
\end{table}
Therefore, controlled $U_{f_i}$ or $G$ is activated if and only if the state of the corresponding control qubit is $|1\>$. For instance, if $0<i\leq 2^{t-1}$, we have
\begin{equation*}
    |c\>\sum_j\alpha_j|j\>|u\>\xlongrightarrow{\text{~Controlled~} U_{f_i}} |c\>U_{f_i}^{c_0}(\sum_j\alpha_j|j\>|u\>).
\end{equation*}

\subsubsection{State Evolution.}
We now examine the state evolution in \ref{alg:ProtocolLoop}. Suppose the initial state of $q_{c1}$, $q_{a1}$ and $q_{d1}$ is
\begin{equation*}
    |\varphi\> = \frac{1}{\sqrt{T}}\sum_c\sum_j\alpha_{c,j}|j\>|\vec{0}\>.
\end{equation*}
Then its evolution can be summarised in Table \ref{Table:StateGrover}, in which we only illustrate the part on $q_{a1}$ and $q_{d1}$. It is worth noting that  Step \ref{line:LTerminating1} and Step \ref{line:LTerminating2} will never change the state of $q_{c1}$, $q_{a1}$, $q_{d1}$ and $q_{c2}$, $q_{a2}$, $q_{d2}$. This is because during the tests, all these qubits are preserved and only new constructed test states are sent to Bob.
Therefore, the controlled Grover iteration is actually realized in procedure \ref{alg:ProtocolLoop}. Moreover, if Bob uses a trivial strategy in which all $f_i=f$, Eq.\eqref{eq:controlG} is implemented. For nontrivial strategies, see Section \ref{Sec:BobStrategyInner}.
\begin{table}[t]
\center
\begin{tabular}{l|l|l}
\hline
  \multirow{2}{*}{Step} & \multicolumn{2}{c}{State}\\
   & Operators activated & Operators not activated\\
  \hline
  Step 1 & $|c\>\sum_j \alpha_{j}|j\>|\vec{0}\>$ & $|c\>\sum_j \alpha_{j}|j\>|\vec{0}\>$ \\
  Step \ref{line:L1} Alice $U_D(y)$ & $|c\>\sum_j \alpha_{j}|j\>|d_{j\oplus y}\>$ & $|c\>\sum_j \alpha_{j}|j\>|d_{j\oplus y}\>$ \\
  Step \ref{line:L1} Bob $U_f$ & $|c\>\sum_j (-1)^{f(d_{j\oplus y})}\alpha_{j}|j\>|d_{j\oplus y}\>$ & $|c\>\sum_j \alpha_{j}|j\>|d_{j\oplus y}\>$ \\
  Step \ref{line:L2} Alice $U_D(y)$ & $|c\>\sum_j (-1)^{f(d_{j\oplus y})}\alpha_{j}|j\>|\vec{0}\>$ & $|c\>\sum_j \alpha_{j}|j\>|\vec{0}\>$ \\
  Step \ref{line:L2} Bob $G$ & $|c\>\sum_j (-1)^{f(d_{j\oplus y})}\alpha_{j}(G|j\>)|\vec{0}\>$ & $|c\>\sum_j \alpha_{j}|j\>|\vec{0}\>$ \\
  \hline
\end{tabular}
\caption{State evolution for $|c\>\sum_j\alpha_j|j\>|\vec{0}\>$ in procedure \ref{alg:ProtocolLoop}. Here ``Operators activated'' means $|c\>$ and $i$ activate controlled operators.\label{Table:StateGrover}}
\end{table}

\subsection{Alice's Tests}\label{Sec:AliceTests}

\begin{procedure}[t]
    \SetKwData{Left}{left}\SetKwData{This}{this}\SetKwData{Up}{up}
    \SetKwFunction{Union}{Union}\SetKwFunction{FindCompress}{FindCompress}
    \SetKwInOut{Input}{Input}\SetKwInOut{Output}{Output}
    \Output{ ``Dishonesty detected'', if Bob's dishonesty is detected.}
    \Begin{
    Alice generates $\mu<\nu\in\{0,1\}^k$, $c\in\{0,1\}^t$, $m\in\{0,1,\cdots,n-1\}$, $x\in \{0,1\}^n$, and $b\in\{0,1\}$ uniformly at random\;
    Alice prepares $|\Phi\> = |c\>_{q_c}\otimes U_t(m,x,b)|0\>^{\otimes (n+k)}_{q_{a}, q_{d}}$ {  on new control qubits $q_c$, address qubits $q_a$, and data qubits $q_d$}\;\label{line:T1}
    Bob applies $U_{f_i}$ on $q_{a}$ and $q_{d}$\;
    Alice applies $U_t(m,x,b)^\dag$ on $q_{a}$ and $q_{d}$, obtaining the state $$|\Phi_1\> =  |c\>_{q_c}\otimes U_t(m,x,b)^\dag U_{f_i} U_t(m,x,b)|0\>^{\otimes (n+k)};$$\label{line:T9}\\
    \vspace{-1em}
    Alice and Bob repeat Step \ref{line:T1} to Step \ref{line:T9} to get $|\Phi_2\>$\;
    Alice measures all address and data qubits of $|\Phi_1\>$, and $|\Phi_2\>$, according to the basis $\{|0\>,|1\>\}$. Let the outcomes be $v$ and $w$, respectively, both in $\{0,1\}^{n+k}$\;
    {  If $v_0\neq w_0$, or $v_j=1$, or $w_j=1$ for any $j>0$, \Return{``Dishonesty detected''}\;\label{line:TTerminating}}
    }
    \caption{TestBob1($i$)}\label{alg:TestACD1}
\end{procedure}

Finally, we present the two test procedures called in Algorithm~\ref{alg:ProtocolInner} and procedure \ref{alg:ProtocolLoop} to complete the picture of our protocol. Since they are similar, we only show  procedure \ref{alg:TestACD1} in this subsection. The differences between it and \ref{alg:TestACDG} are briefly shown in Figure \ref{Fig:TestDifference}. The detailed descriptions of \ref{alg:TestACDG} are given Appendix \ref{Apd:TestACD2ACDG}.
\begin{figure}[t]
  \centering
  \fbox{
  \parbox{\textwidth}{
  \begin{itemize}
    \item Ancilla qubits $q_{g1}$ and $q_{g2}$ are involved.
    \item A controlled swap test is employed to test whether $|\Phi_1\>$ and $|\Phi_2\>$ are same.
    \item Step \ref{line:TTerminating}: The condition is a bit different.
  \end{itemize}
  }}
  \caption{The difference of \ref{alg:TestACDG} to \ref{alg:TestACD1}. See Appendix \ref{Apd:TestACD2ACDG} for details.}\label{Fig:TestDifference}
\end{figure}
Some details of  procedure \ref{alg:TestACD1} are described as follows:
\begin{itemize}
  \item { $\mu<\nu$ means the binary number represented by $\mu$ is smaller than that represented by $\nu$.}
  \item {  In this test, the state $|c\>$ on control qubits is not checked. It is introduced here, only because originally the states on control qubits are involved during the computation.}
  \item $U_t(m,x,b) = U_{\mathit{SWAP}(0,m)}Z(x)X_0^{b}V(\mu,\nu)(W\otimes I_d)$, {  where $W=H^{\otimes n}$ and $H$ is the Hadamard gate.}
  \item $V(\mu,\nu)$ writes $\mu$ and $\nu$ into $|+\>^{\otimes n}_{q_{a}}|0\>^{\otimes k}_{q_{d}}$. It consists of at most $k$ CNOT gates, where the control qubits are the first address qubit, and the target qubits range over all data qubits, where the first address qubit serves as the control, and all data qubits are the target. To be specific, $$V(\mu,\nu)|0\>|\xi\>|\tau\> = |0\>|\xi\>|\tau\oplus\mu\>,$$ $$V(\mu,\nu)|1\>|\xi\>|\tau\> = |1\>|\xi\>|\tau\oplus\nu\>,$$
      for any $\xi\in\{0,1\}^{n-1}$ and $\tau\in\{0,1\}^k$. For simplicity, $V(\mu,\nu)$ will be abbreviated to $V$ at some places.
  \item $U_{\mathit{SWAP}(0,m)}$ swaps the states of address qubits number 0 and number $m$.
  \item $Z(x)$ consists of a sequence of Pauli $Z$ gates which act on address qubit $j$ if and only if $x_j=1$.
  \item $X_0$ denotes Pauli $X$ gate acting on the first address qubit.
  \item Since all the component operators in $U_t(m,x,b)$ are self-adjoint, so is $U_t(m,x,b)$; that is, $U_t(m,x,b)^\dag = U_t(m,x,b)$.
\end{itemize}

\section{Execution of the Protocol}\label{Sec:ProtocolExecution}

\subsection{Test Rounds}

To better understand our protocol, let us show in this section how it is executed. We first see how dishonest Bob cannot pass tests with a high probability.
\begin{definition}\label{defi:TestRound}
\begin{enumerate}
  \item One execution of procedure \ref{alg:TestACD1} or \ref{alg:TestACDG} is called \textit{a test round}.
  \item Correspondingly, one execution from Step \ref{line:L1} to Step \ref{line:L1r} in procedure \ref{alg:ProtocolLoop} or one execution from Step \ref{line:B8} to Step \ref{line:B10} in Algorithm \ref{alg:ProtocolInner} is called a computational round.
\end{enumerate}
\end{definition}

One possible state evolution in procedure \ref{alg:TestACD1} is given in Table \ref{table:StateTest}, where the post-measurement state is assumed to be $|0\>|+\>^{\otimes n-1}|\mu\>$ or $|0\>|\gamma\>|\mu\>$ for some $\gamma\in\{0,1\}^{n-1}$. We can see that if Bob is honest, he can always passes \ref{alg:TestACD1}.
On the other hand, we present two examples to illustrate how Bob's attack can be detected (A detailed analysis is postponed to Section \ref{Sec:PrivacyAliceAnalysis}).
\begin{table}[t]
\center
\begin{tabular}{c|l|l|l}
\hline
   &  Honest Actions  &  Measurements  on $q_d$ & Measurements  on  $q_a$ and $q_d$\\
  \hline
  Alice $\ra$ Bob & $|\psi_{0,0,0}(\mu,\nu)\>$ & $|\psi_{0,0,0}(\mu,\nu)\>$ & $|\psi_{0,0,0}(\mu,\nu)\>$\\
  Bob $\ra$ Alice &$|\psi_{0,0,0}(\mu,\nu)\>$ & $|0\>|+\>^{\otimes n-1}|\mu\>$ & $|0\>|\gamma\>|\mu\>$\\
  $V(\mu,\nu)$ & $|+\>^{\otimes n}|0\>^{\otimes k}$ & $|0\>|+\>^{\otimes n+k-1}$  & $|0\>|\gamma\>|0\>^{\otimes k}$ \\
  $W\otimes I_d$ & $|0\>^{\otimes n+k}$ & $|+\>|0\>^{\otimes n+k-1}$& $|+\>|\omega\>|0\>^{\otimes k}$ \\
  \hline
\end{tabular}
\caption{One possible situation of state evolution in Procedure \ref{alg:TestACD1}, where $\gamma\in\{0,1\}^{n-1}$, and $\omega\in\{+,-\}^{n-1}$. Moreover it is assumed that $f(\mu)=f(\nu)$ and the test state is $|\psi_{0,0,0}(\mu,\nu)\>$. The control qubits are omitted. Unitary operators $U_{\mathit{SWAP(0,0)}}$, $X_0^0$, $Z(0)$ are omitted as well, as they are all identity operators here.\label{table:StateTest}}
\end{table}
First, we assume that Bob performs measurements only on the data qubits.
\begin{example}\label{exam:Test1}
    Suppose in a test round of procedure \ref{alg:TestACD1}, Alice first sends a test state $|\psi_{0,0,0}(\mu,\nu)\>$ to Bob, where $f(\mu)=f(\nu)$. Bob performs measurements on the data qubits, gets post-measurement state $|0\>|+\>^{\otimes n-1}|\mu\>$, and then sends it back to Alice. Secondly, Alice sends the same test state $|\psi_{0,0,0}(\mu,\nu)\>$ to Bob. Then, as illustrated in Table \ref{table:TestProbability},
    \begin{enumerate}
      \item if Bob does not perform measurements on the second test state, it will be detected  by the condition $v_0\neq w_0$ at Step \ref{line:TTerminating} with probability 0.5.
      \item if Bob performs measurements on the data qubits, then with probability 0.5, the post-measurement state is $|0\>|+\>^{\otimes n-1}|\mu\>$, and with the same probability, it is $|1\>|+\>^{\otimes n-1}|\nu\>$. For each case, it will be detected by the condition $v_0\neq w_0$ at Step \ref{line:TTerminating} with probability 0.5.
    \end{enumerate}
    \end{example}
\begin{table}
    \center
    \begin{tabular}{cc|c}
    \hline
        \multicolumn{2}{c|}{States after Bob's actions} & Probability at Step \ref{line:TTerminating}\\
        State \#1 & State \#2 &  \\\hline
        \multirow{3}{*}{$|0\>|+\>^{\otimes n-1}|\mu\>$} & {$|\psi_{0,0,0}(\mu,\nu)\>$} & 0.5\\
         & $|0\>|+\>^{\otimes n-1}|\mu\>$ & 0.5\\
         & {$|1\>|+\>^{\otimes n-1}|\nu\>$} & 0.5 \\
    \hline
    \end{tabular}
    \caption{Probabilities to detect Bob's attacks in Example \ref{exam:Test1}. ``State \#1'' (resp. ``State \#2'') stands for the first (resp. second) test state.}\label{table:TestProbability}
\end{table}

One situation not mentioned in the above example is that Bob first performs honestly, and then attacks on the second test state. But it is essential the same as the above example.

Another attack that Bob may take is to perform measurements on both the address and data qubits.
\begin{example}\label{exam:StateEvolutionInTest}
    Similarly to Example \ref{exam:Test1}, Bob measures the address and data qubits for the first test state, gets post-measurement state $|0\>|\gamma\>|\mu\>$, and then sends it back to Alice. Then for the second test state, no matter what Bob does, it can be detected by one or both of the two conditions at Step \ref{line:TTerminating} with probability $\approx 1$.
\end{example}

In the above two examples, we assumed that the test state is $|\psi_{0,0,0}(\mu,\nu)\>$ for simplicity. Other test states are similar.
Furthermore, it was assumed that Bob directly sends the post-measurement state to Alice. Indeed, he can construct a new state and send it to Alice. This case can be detected as well; see Section \ref{Sec:PrivacyAliceRecovery}.

\subsection{Testing or Computing}

The design idea of tests is to guarantee that Bob cannot know whether he is dealing with a test state or a computational state. In Algorithm \ref{alg:ProtocolInner} and procedure \ref{alg:ProtocolLoop}, the order of test and computational rounds are decided by the random number $r$. So it is clear that Bob does not know what states he is dealing with. Every time he receives a quantum state, it may be a test state (with probability at least $p$). Figure \ref{Fig:BobView} shows a flowchart that illustrates Bob's view for the $T$ calls of controlled Grover iterations in Algorithm \ref{alg:ProtocolInner}.

\begin{figure}
  \centering
  \includegraphics[width=8cm]{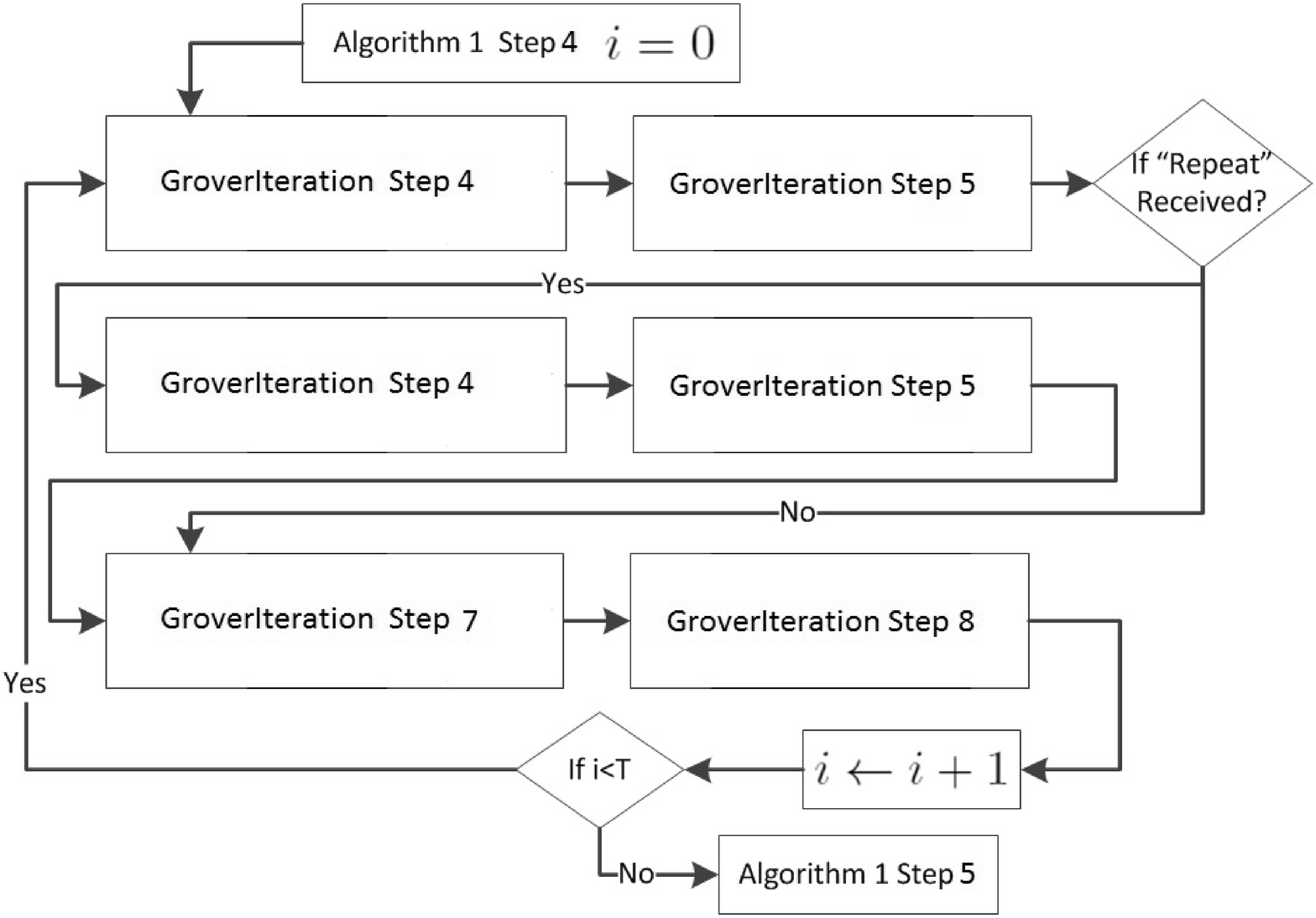}
  \caption{Bob's view for the $T$ calls of controlled Grover iterations in Algorithm \ref{alg:ProtocolInner}. { Indeed, in Alice's view, the first (resp. second) appearance of Steps 4 and 5 may be a test at Step 3 (resp. Step 6) when Alice employs a test in \ref{alg:ProtocolLoop} according to $r\leq p$ (resp. $p<r\leq 2p$). }}\label{Fig:BobView}
\end{figure}

\subsection{Bob's Strategy}\label{Sec:BobStrategyInner}

As said in Subsection \ref{ideas-0}, Bob's strategy is to add noises to $f$ (thus applying $f_i$ instead) in procedure \ref{alg:ProtocolLoop}, which cancel each other if Alice follows the protocol honestly. A detailed analysis will be given in Section \ref{Sec:PrivacyBob}. Here we observe:
\begin{equation}\label{eq:BobStrategy1}
    \bar{G}U_DU_fU_D~\bar{G}U_DI_{a,d}U_D~\bar{G}U_DU_fU_D=\bar{G},
\end{equation}
where $U_D$ is short for $U_D(y)$. This equation indeed represents the unitary operators after three iterations with $f_{i} = f$, $f_{i+1} = h$, and $f_{i+2}=f$, where $h$ is the function corresponding to $I_{a,d}$ and $h(\mu)=0$ for all $\mu\in\{0,1\}^k$. So, a sequence of $f$, $h$, $f$, $h$ or $h$, $f$, $h$, $f$ leads to the identity operator, meaning that Bob does nothing.
One step further from Eq. \eqref{eq:BobStrategy1}, we have:
\begin{equation}\label{eq:BobStrategy1+}
    \prod_{i=1}^j (\bar{G}U_DU_{f_i}U_D)~\bar{G}U_DI_{a,d}U_D~\prod_{i=1}^j (\bar{G}U_DU_{f_{j+1-i}}U_D)=\bar{G},
\end{equation}
for all $j$ and all functions $f_1,\cdots, f_j$ on $D$.
Another useful observation is
\begin{equation}\label{eq:BobStrategy2}
    \bar{G}U_DI_{a,d}U_D~\bar{G}U_DI_{a,d}U_D=I_{a,d}.
\end{equation}
This means two repetitions of $h$ do nothing.

We can construct strategies for privacy preserving based on Eqs. \eqref{eq:BobStrategy1}, \eqref{eq:BobStrategy1+} and \eqref{eq:BobStrategy2}. Let us first see a simple example.
\begin{example}
    Suppose Bob wants to run Algorithm \ref{alg:ProtocolInner} with $T_0=8$ loops. He adds one control qubit to make $T=16$. Then there are four control qubits, denoted by $C_0$, $C_1$, $C_2$, $C_3$. Bob chooses $C_1$ as a confusing qubit, and add noise to the functions corresponding to $C_0$. In detail,
    \begin{itemize}
      \item $C_0$: $h$, $h$, $f$, $f$, $h$, $h$, $f$, $f$,
      \item $C_1$: $f'$, $h$, $f'$, $h$,
      \item $C_2$: $f$, $f$,
      \item $C_3$: $f$,
    \end{itemize}
    where $f'$ are functions different from $f$. At last the Fourier transformation is performed only on $C_0$, $C_2$, $C_3$.
\end{example}

We now formally define the notion of Bob's strategy.
\begin{definition}\label{defi:BobStrategy}
    Suppose in Algorithm \ref{alg:ProtocolInner} there are $t=\log T$ control qubits $C_0$, $\cdots$, $C_{t-1}$. We say a sequence of function $f_0$, $\cdots$, $f_{T-1}$ is a strategy $\mathcal{S}$ for computing $f(D)$, if
    \begin{itemize}
      \item it is trivial, i.e. all $f_i = f$,
      \item or the following conditions are satisfied:
        \begin{itemize}
          \item One control qubit $C_u$ with $u<t-1$ is chosen to be the confusing qubit.
          \item Based on Eq. \eqref{eq:BobStrategy1}, Eq. \eqref{eq:BobStrategy1+} and Eq. \eqref{eq:BobStrategy2}, noises are added between the functions corresponding to $C_0,\cdots,C_{u}$.
          \item For $C_w$ with $w<u$, the effect of its corresponding functions is equivalent to that of $2^{t-w-1}$ repetitions of $f$.
          \item For $C_u$, the effect of its corresponding functions is the identity.
        \end{itemize}
    \end{itemize}
\end{definition}

In conclusion, a strategy realizes Eq. \eqref{eq:controlG} or
\begin{align}
    &\frac{1}{\sqrt{T}}\sum_{c} |c\>G^{c'} |\Psi_0\> =\nonumber\\
    &\frac{1}{\sqrt{T}}\sum_{c} |c\> \otimes (G^{c_0 2^{t-2}+c_1 2^{t-3}+\cdots +c_{u-1}2^{t-u-1}+c_{u+1}2^{t-u-2}+\cdots +c_{t-1} } |\Psi_0\>),\label{eq:controlG2}
\end{align}
{ where $c = \sum c_j 2^{t-j-1}$, and $c' = c_0 2^{t-2}+c_1 2^{t-3}+\cdots +c_{u-1}2^{t-u-1}+c_{u+1}2^{t-u-2}+\cdots +c_{t-1}$.}
If a confusing qubit is added, the final measurement is performed on the control qubits except $C_u$ to get $\theta$.

\section{Correctness for Honest Parties}\label{Sec:Correctness}
Now we start to prove the correctness of Algorithm \ref{alg:ProtocolInner}. In this section, we consider the simplest case where Alice and Bob are both honest.

The discussion in Sections \ref{Sec:GroverIteration} and \ref{Sec:BobStrategyInner} showed that in Algorithm \ref{alg:ProtocolInner}, the $T$ calls of procedure \ref{alg:ProtocolLoop} realize controlled Grover iterations Eq.\eqref{eq:controlG} or Eq.\eqref{eq:controlG2}, depending on which strategy (see Definition \ref{defi:BobStrategy}) Bob employs. So Algorithm \ref{alg:ProtocolInner} executes the  quantum counting algorithm \cite{BrassardHT1998} twice, independently on $q_{c1}$, $q_{a1}$, $q_{d1}$, $q_{g1}$, and $q_{c2}$, $q_{a2}$, $q_{d2}$, $q_{g2}$. {  Then for Problem \ref{Prob:count}, we directly have the following result.
\begin{theorem}(\cite[Theorem 5, Theorem 6]{BrassardHT1998})\label{thm:ResultAccuracy}
    In Algorithm \ref{alg:ProtocolInner}, if Alice and Bob are both honest and Bob employs a trivial strategy, then for $i\in\{1,2\}$,
    \begin{equation}\label{eq:Result}
        |s_i-s|<\frac{2\pi}{T}\sqrt{s}+\frac{\pi^2}{T^2}, \forall i\in\{1,2\},
    \end{equation}
 holds with probability at least $\frac{8}{\pi^2}>0.8$, where $s=f(D)=\frac{1}{N}\sum_j f(d_j)$ is the correct answer.
\end{theorem}

Therefore, by setting $T>100/\sqrt{s_{\min}}$ for a trivial strategy or $T>200/\sqrt{s_{\min}}$ for a nontrivial strategy, we immediately obtain:
\begin{corollary}\label{cor:ResultAccuracy}
    In Algorithm \ref{alg:ProtocolInner}, if Alice and Bob are both honest, then for $i\in\{1,2\}$,
    \begin{equation*}
        \begin{cases}
            |\frac{s_i}{s}-1|<\frac{2\pi}{100}\sqrt{\frac{s_{\min}}{s}}+\frac{\pi^2}{10^4}\frac{s_{\min}}{s}<0.07, & s\geq s_{\min},\\
            \frac{|s_i-s|}{s_{\min}}<\frac{2\pi}{100}\sqrt{\frac{s}{s_{\min}}}+\frac{\pi^2}{10^4}<0.07, & s< s_{\min},
        \end{cases}
    \end{equation*}
 holds with probability at least $\frac{8}{\pi^2}>0.8$, where $s=f(D)=\frac{1}{N}\sum_j f(d_j)$ is the correct answer, and ${s_{\min}}$ is the preset threshold of supports.
\end{corollary}
}

This corollary gives the relative error and success probability for Bob. Since usually $s_{\min}$ is set to be a constant, say 0.2, 0.1 or 0.01, the number $T$ of iterations does not depend on the size of the database.

\section{Protecting Alice's Privacy}\label{Sec:PrivacyAlice}
In this section, we continue to prove correctness of the protocol and show how it can protect Alice's privacy. Only procedure \ref{alg:TestACD1} is considered, and the results for procedure \ref{alg:TestACDG} are similar and thus omitted.

Dishonest Bob may employ attacks to read information from  $q_{c1}$, $q_{a1}$, $q_{d1}$, $q_{g1}$, and/or $q_{c2}$, $q_{a2}$, $q_{d2}$, $q_{g2}$. His attacks can be classified according to the number of rounds/iterations  that these attacks cost.

\subsection{One-round Attacks}
\begin{definition}[One-round attack]
    A one-round attack consists of one or more successive steps of the following:
    \begin{enumerate}
      \item Bob sends some qubits to Alice,
      \item Alice applies $U_D$, and sends these qubits to Bob,
      \item Bob's actions,
      \item Bob sends some qubits to Alice.
    \end{enumerate}
\end{definition}

For instance, a one-round attack may consist of Steps 1-2, Steps 2-4, or Steps 1-4.  In this subsection, we list some one-round attacks which leak information.

The following are several typical one-round attacks:
\begin{example}[\textbf{Attack1}]\label{exam:MPattack1}
    Bob sets $q_a$ to be a single address $|i\>$ and $q_d$ to be blank $|\vec{0}\>$. Then he sends $q_c$, $q_a$, $q_d$ to Alice. After Alice applies $U_D$, Bob receives $|i\>|d_i\>$. Finally, he can measure $q_d$ to get $d_i$.
\end{example}

\begin{example}[\textbf{Attack2}]\label{exam:MPattack2}
    After Bob receives $\frac{1}{\sqrt{T}}\sum_c\sum_j\alpha_{c,j}|j\oplus y\>|d_{j\oplus y}\>$, he performs measurements (1) on $q_d$ to get some $d\in D$, or (2) on $q_a$ and $q_d$ to get $d_{j\oplus y}$.
\end{example}

Besides directly reading data from $q_a$ and $q_d$ by measurements, Bob may use some unitary gates (e.g. CNOT) to copy data on additional blank qubits, say $q_e$. Then he can read information from $q_e$ later.
\begin{example}[\textbf{Attack3}]\label{exam:MPattack3}
    After Bob receives $\frac{1}{\sqrt{T}}\sum_c\sum_j\alpha_{c,j}|j\oplus y\>|d_{j\oplus y}\>$, he add ancilla qubits and performs unitary operators to store data, i.e. $$\frac{1}{\sqrt{T}}\sum_c\sum_j\alpha_{c,j}|j\oplus y\>|d_{j\oplus y}\>|e_{j\oplus y}\>.$$
\end{example}

\subsection{Detection of One-round Attacks}\label{sec:MidAttackStrongDetection}
We first show that \ref{alg:TestACD1} can detect Bob's one-round attacks. The main assumption here is:
\begin{itemize}
  \item Whenever Bob tries to attack, he believes that he cannot distinguish the following two situations from each other:
    \begin{enumerate}
      \item He is dealing with a test state.
      \item He is dealing with a computational state.
    \end{enumerate}
\end{itemize}
This assumption is reasonable because no one will cheat if he knows that he is dealing with a test state, which carries no useful information about database $D$.

\subsubsection{Example Attacks}
Now let us see what happens if Bob cheats in a test. As the starting point, we focus on the attacks in Examples \ref{exam:MPattack1}, \ref{exam:MPattack2} and \ref{exam:MPattack3}, since they are simplest and typical ones.

\begin{lemma}[\textbf{Attack1}]\label{lem:Attack1C}
    Suppose, in a test, it is Bob's turn to send back computational qubits to Alice. Bob prepares $|\varphi\>|a\>|d\>$ in Procedure \ref{alg:TestACD1},
    where $|\varphi\>\in\hs_c$, $a\in\{0,1\}^n$ and $d\in\{0,1\}^k$, and sends this state to Alice. Then it can be detected with probability at least $1-\frac{1}{N}$.
\end{lemma}

\begin{lemma}[\textbf{Attack1}]\label{lem:Attack1M}
    Suppose Bob successfully sends $|\varphi\>|a\>|d\>$ to Alice. If this communication is followed by a test (Procedure \ref{alg:TestACD1}) and he performs measurements on $q_a$ and $q_d$ (resp. only a measurement on $q_d$) [in order to get private information], then it will be detected with probability at least $1-\frac{1}{N}$ (resp. $\frac{1}{2}$).
\end{lemma}

\begin{lemma}[\textbf{Attack2}]\label{lem:Attack2}
    Suppose Bob performs measurements on $q_a$ and $q_d$ (resp. only a measurement on $q_d$) in Procedure \ref{alg:TestACD1}. Then it will be detected with probability at least $1-\frac{1}{N}$ (resp. $\frac{1}{2}$).
\end{lemma}
Lemma \ref{lem:Attack2} is essentially the same as Lemma \ref{lem:Attack1M}, since in both of them the same measurements are performed on a test state.
But the analysis of \textbf{Attack3} is much more complicated.
\begin{lemma}[\textbf{Attack3}]\label{lem:Attack3}
    Suppose in Procedure \ref{alg:TestACD1}, Bob adds $q_g$ to $q_a$ and $q_d$ and uses a unitary operator $E$ to entangle $q_g$ to $q_a$ and $q_d$:
    \begin{equation}
        E|i\>|d\>|{0}\> = |i\>|d\>|\lambda_{i,d}\>,
    \end{equation}
    where $|\lambda_{i,d}\>$ is a pure state of $q_g$. In order to read information, $|\lambda_{i,d}\>$ should vary for $i,d$. Then it will be detected with a positive probability $P_{DET}$ depending on $E$. In particular,
      \begin{enumerate}
      \item if $E|i\>|d\>|{0}\> = |i\>|d\>|i\>|d\>,$ then $P_{DET}\geq 1-\frac{1}{N}$.
      \item if $E|i\>|d\>|{0}\> = |i\>|d\>|d\>$, then $P_{DET}=0.5$.
    \end{enumerate}
\end{lemma}

\subsubsection{General Attacks}
Now let us consider general attacks. The following theorem identifies all of Bob's actions that enable him to pass Alice's tests.
\begin{theorem}\label{thm:strongPassMeasurement}
    Suppose Bob applies a super-operator $\E = \sum_j E_j\circ E_j^\dag$ on $q_a$, $q_d$ and blank $q_g$ in a round of \ref{alg:TestACD1}. If it always passes the test, $\E$can be written as
    \begin{equation*}
        \E = U\circ U^\dag \otimes \E_g,
    \end{equation*}
    where $U$ is a unitary operator on $q_a$ and $q_d$, and $\E_g$ is a super-operator on $q_g$.
\end{theorem}

Its implication to privacy is the following:

\begin{corollary}\label{cor:AlwaysPassTest}
    If Bob wants to always pass the tests, he cannot read any information from $q_a$ and $q_d$ by one-round attacks.
\end{corollary}

Note that in procedure \ref{alg:ProtocolLoop}, we only add tests around Step \ref{line:L1}. One question directly arises: what happens if Bob attacks at Step \ref{line:L2}? The following lemma answers this question.
\begin{lemma}\label{lem:RemoveingTest2}
     At Step \ref{line:L2}, if Bob performs measurements to read information or sends a special state for future attacks at Step \ref{line:L1}, then it can be detected by procedure \ref{alg:TestACD1}.
\end{lemma}

\subsection{Impossibility of Recovery}\label{Sec:PrivacyAliceRecovery}
In this subsection, we further show that Bob cannot distinguish the test states. Therefore, he cannot recover his measurement even if he finds that he is dealing with a test state.

The impossibility of distinguishability is based on the following observation:
\begin{lemma}\label{lem:TestStateBasis}
    The test set $\{|\psi_{m,x,b}(\mu,\nu)\>\}$ can be decomposed into the union of $n(2^k-1)$ disjoint bases:
    \begin{equation*}
        \{|\psi_{m,x,b}(\mu,\nu)\>\} = B_1\cup\cdots\cup B_{n(2^k-1)},
    \end{equation*}
    where $B_i\cap B_j=\emptyset$, and $B_i$ is a orthogonal basis of the Hilbert space of $q_a$ and $q_d$, for all  $i\neq j$.
\end{lemma}

The above lemma then implies that Bob cannot distinguish all of the test states.

\begin{lemma}\label{lem:TestStateDis}
    Suppose Bob tries to use  measurement $\{M_v\}$ on $q_a$ and $q_d$ to find which specific test state Alice sends. Then the correct probability is \begin{equation}
        \Pr(m,x,b,\mu,\nu|M_v) \leq \frac{1}{n(2^k-1)}.
    \end{equation}
    In other words, if the measurement outcome is $v$, then the probability that the state is $|\psi_{m,x,b}(\mu,\nu)\>$ is at most $\frac{1}{n(2^k-1)}$. More generally,
    \begin{equation}
        \Pr(B|M_v) = \sum_{|\psi_{m,x,b}(\mu,\nu)\>\in B}\Pr(m,x,b,\mu,\nu|M_v) \leq \frac{1}{n(2^k-1)},
    \end{equation}
    where $B$ is an orthogonal basis as in Lemma \ref{lem:TestStateBasis}.
\end{lemma}

Now we can present the main theorem in this subsection.
\begin{theorem}\label{thm:Recovery}
    Suppose Bob uses a measurement $\{M_v\}$ to read information in a test round, and sends a new state $|\psi_{m',x',b'}(\mu',\nu')\>$ (based on the measurement results) back to Alice instead. Then, the expected success probability that he passes the test is at most $\frac{1}{4}+\frac{3}{4n(2^k-1)}$.
\end{theorem}

The above theorem ensures that Bob cannot recover the test states after attacks. Recall from Theorem \ref{thm:strongPassMeasurement} that if Bob directly sends states back to Alice after his attacks, it will be detected. So, these two theorems together warrant that once Bob wants to read private information from Alice through one-round attacks, it will be detected.

\subsection{Multi-round Attacks}\label{Sec:MultiRoundAttack}
We have discussed one-round attacks in the last subsection. In this subsection, we further consider multi-round attacks.
A multi-round attack will finish in more than one calls of procedure \ref{alg:ProtocolLoop}. More precisely, we have:
\begin{definition}
    A multi-round attack consists of the following steps:
    \begin{enumerate}
      \item Bob sends some qubits to Alice,
      \item Several calls of Procedure \ref{alg:ProtocolLoop} are executed,
      \item Bob performs measurements to read information after receiving qubits from Alice.
    \end{enumerate}
\end{definition}

We will see in Sections \ref{Sec:MultiRoundAttack} and \ref{Sec:PrivacyAliceMRAttack} that multi-round attacks can actually be ignored, since (1) they can hardly leak information, and (2) they are very hard to be detected. But here let us see two typical multi-round attacks:
\begin{example}\label{exam:MRAttack1}
    Bob employs function $f(x) = \delta(x,d)$ as the target function, where $\delta(x,d) = 1$ if and only if $x=d$. Then Bob runs the protocol honestly to find whether $d\in D$.
\end{example}

This function discloses the information whether $d\in D$ and can be treated as an attack, since we only allow target functions to be maps indicating the inclusion relation $\subseteq$. For this kind of attacks, Alice can construct tests to detect it with a certain probability, although this probability is extremely low (see Appendix \ref{Apd:Attackx=d}).

Fortunately, we can ignore this attack because (1) if $\supp(d)$ is high,  then the information $d\in D$ is no longer private information when mining association rules or decision trees, and (2) if $\supp(d)$ is low, the result can be hardly derived from Algorithm \ref{alg:ProtocolInner} (see Section \ref{Sec:PrivacyAliceAnalysis} for more details).

Another attack focuses on more specific information.
\begin{example}\label{exam:MRAttack2}
    Bob employs Oracle
    \begin{equation*}
        U_f|j\>|d_j\> =(-1)^{\delta(j,i)g(d_j)}|j\>|d_j\>,
    \end{equation*}
    as the target function in Algorithm \ref{alg:ProtocolInner}, where $g(x) = 1$ if and only if $x\subseteq d$. One alternative attack is $g(x) = \delta(x,d)$.
\end{example}
As shown in the next lemma, this kind of attacks is very hard to detect.
\begin{lemma}\label{lem:MultiRoundAttackImpossibility}
    Suppose Bob acts as in Example \ref{exam:MRAttack2}, and Alice only employs tests based on state comparison in a single round (the tests are not restricted to those in this paper). It can be detected in a single round with probability at most $O(\frac{1}{N})=O(\frac{1}{2^n})$. Furthermore, Bob can passes all tests in one execution of Algorithm \ref{alg:ProtocolInner} with probability approximately 1.
\end{lemma}

Fortunately, since the database $O_D$ is modified to be $U_D(y)$, Example \ref{exam:MRAttack2} is reduced to Example \ref{exam:MRAttack1} finally. Indeed, Bob wants finally to get the exact form of $d_j$ by employing attacks in Example \ref{exam:MRAttack2}. But since $d_j$ is changed to be $d_{j\oplus y}$ and Bob does not know $y$, Bob only gets information $d\in D$ for some $d$. This is exactly the case of Example \ref{exam:MRAttack1}.

\subsection{Attacks on $q_c$ and $q_g$}\label{Sec:PrivacyAliceQcg}
In the previous subsections, we only consider attacks on $q_a$ or on $q_d$. In this subsection, we analyse attacks on the two parts jointly.

First, by the following observations, we can see that $q_c$ will not introduce further information leakage:
\begin{itemize}
  \item There is no information about $D$ on $q_c$.
  \item Bob cannot verify whether the current round is a test round by measurements on $q_c$. This is because no matter Bob performs measurements on $q_c$ in a test round or an original round, the outcome distributions are the same, i.e.,
    \begin{equation}\label{eq:controlTestOriginalRound}
        \Pr(c=i|\text{test round})=\frac{1}{T} = \Pr(c=i|\text{original round}).
    \end{equation}
\end{itemize}

Second, for $q_g$, since $q_g$ is entangled to $q_d$, attacks on $q_g$ are the same as the attacks on $q_d$, which we have already analysed.

\subsection{Privacy Analysis}\label{Sec:PrivacyAliceAnalysis}
Now, we are able to analyse the privacy level of the entire protocol. Let us first examine information disclosed by one-round attacks, and then give the privacy analysis for multi-round attacks.

\subsubsection{The Entire Database}\label{Sec:PrivacyAliceEntireD}
In this subsection, we analyse the privacy of the entire database; that is, how much of $D$ will be disclosed if Bob is dishonest? Consider the following:
\begin{example}\label{exam:StoryEntireDatabase}
    Suppose Alice is a data provider, who sells data, and Bob is a costumer, who wants to buy some access to the data from Alice. Alice wants to keep her data private, as she wants to sell it to other costumers. Bob wants to keep his research private, as his research outcome may bring outcomes. So he will not send the function $f$ to Alice.
\end{example}
In this example, Alice tries to preserve the entire database. Then how many transactions will be disclosed in our protocol:

\textbf{Case 1}. Bob is honest: He exactly follows the protocol. Before measurements in Algorithm \ref{alg:ProtocolInner}, due to the quantum counting algorithm \cite{BrassardHT1998}, he holds the state
\begin{equation*}
    \sum_c\sum_j\alpha_{c,j}|c\>|j\>|\vec{0}\>|g_j\> = \frac{1}{\sqrt{T}}\sum_c|c\>(\beta_{c,0}|\varphi_0\>|\vec{0}\>|0\>+\beta_{c,1}|\varphi_1\>|\vec{0}\>|1\>),
\end{equation*}
where $g_j = f(d_{j\oplus y})$, and $\alpha_{c,j}$, $\beta_{c,0}$, $\beta_{c,1}$ are amplitudes. Since honest Bob only performs measurements on the control qubits $q_c$ and ancilla qubit $q_g$ of this state, he only gets the information of $f(D)$. As he knows nothing else, no transactions in $D$ is disclosed.

\textbf{Case 2}. Bob is semi-honest: He may do further computation on the state $|\theta\>|\varphi_j\>|\vec{0}\>|j\>$ with $j=0$ or $j=1$, which he holds after the final computation.
From this state, the information, which Bob can further get by measurements on the address qubits $q_a$, is whether $g_j=f(d_{j\oplus y})$ is 1 or 0 for some $j\in\{0,1\}^n$. Totally he can get this information $g_{j_1}$ and $g_{j_2}$ for two address $j_1$ and $j_2$ randomly generated from measurements in one run of Algorithm \ref{alg:ProtocolInner}, as there are two copies of states in one run. But unfortunately Bob does not know $y$, and Alice changes $y$ in every run of Algorithm \ref{alg:ProtocolInner}. This means the information he gets is useless. In detail, note that $g_j=f(d_{j\oplus y})$ is a random variable dependent on $y$. Since $y$ is chosen uniformly at random,  we have:
\begin{equation*}
    \Pr(g_j=1)=f(D), \forall j.
\end{equation*}
So what Bob can get from $g_j$ is $f(D)$, which he already known from the honest computation. Therefore, Bob disclose no detailed transaction in $D$.

\textbf{Case 3}. Bob is dishonest: He may perform measurements on the state received from Alice at any time. In previous subsections, we already observed that once Bob tries to get information in a test round through measurements, he may have a probability at least 0.5 to be detected. As a consequence, Alice will stop the whole computation. Then the expected number $E_c$ of rounds that Bob can cheat before being detected, may be computed as follows. If Bob's attack happens in the first (resp. second) round of a loop $i$, it will be a test round with probability $p$ (resp. 0.5). So each time Bob tries to get information through one-round attacks, it will be detected with probability at least $0.5p$. Thus, the expected numbers of one-round attacks before being detected is
\begin{equation*}
    E_c \leq \sum_{i\geq 1} i*0.5p*(1-0.5p)^i = \frac{2}{p}-1 = O(1/p).
\end{equation*}
Therefore, dishonest Bob can get at most a constant number of transactions from $D$.

To conclude this section, let us see the advantage of our quantum protocol over a classical method. Usually, a classical data provider will provide a modified database $D'$, generated from $D$ by adding noise into it, to Bob. So, if the quantum protocol is run on $D$, it is not appropriate to compare it with a classical protocol. But the quantum protocol can also run on $D'$, by combining it with a classical one  together (see Section \ref{sec:CombiningCQ}). A comparison of the quantum protocol (combined with a classical one) with a classical protocol is shown in Table \ref{table:PrivacyEntireD}. In this table,  $O(TM)$ is the number of transactions disclosed without protection in the protocol.

\begin{table}
    \center
    \begin{tabular}{c|ccc}
        & Honest Bob & Semi-honest Bob & Dishonest Bob\\\hline
        Quantum Protocol & 0 & $\approx$ 0& $O(\frac{1}{p})$\\
        Quantum Protocol & \multirow{2}{*}{0} & \multirow{2}{*}{$\approx$ 0} & \multirow{2}{*}{$O(TM)$}\\
        without tests & & &\\
        Classical Method & N & N & N \\
    \end{tabular}
    \caption{Number of transactions disclosed in $D$ or $D'$. $p\in (0,1)$ is a constant, $N$ is the size of database, $T$ is the number of iterations in one run of Algorithm \ref{alg:ProtocolInner}, and $M$ is the total number of runs of Algorithm \ref{alg:ProtocolInner}.}\label{table:PrivacyEntireD}
\end{table}

\subsubsection{Multi-round Attack}\label{Sec:PrivacyAliceMRAttack}
In Section \ref{Sec:MultiRoundAttack}, we already mentioned that multi-round attacks can be ignored in our quantum protocols. Now we are ready to give a detailed explanation.

Suppose that Bob employs function $f(x)=\delta(x,d)$ in order to learn whether $d\in D$. He uses this function to run Algorithm \ref{alg:ProtocolInner} several times and get an approximate result $s\approx |\{j:d= d_j\}|/N$. Then there are the following three situations:
\begin{itemize}
  \item $s>s_{\min}$, where constant $s_{\min}$ is the threshold of support. Then this result is not treated to be private, as it cannot be distinguished from that of a \emph{frequent itemset} and thus can be mined by Bob legally. For instance, since $\supp(d) = |\{j:d\subseteq d_j\}|/N\geq |\{j:d= d_j\}|/N\approx s$, Bob can first get $\supp(d)$ legally. Then he computes $\supp(d')$ for possible supersets $d'\supsetneq d$ to approximate $s$.
  \item $s<s_{\min}$, but $s$ is not far from $s_{\min}$. In this case, since parameter $T$ is determined by $s_{\min}$, the results may be not far to $s_{\min}$. For instance, if $T$ is set to be $T>100/\sqrt{s_{\min}}$, then by Corollary \ref{cor:ResultAccuracy} we see that Bob may get result $0.07s_{\min}$ with probability greater than 0.8. So this still cannot be treated to be private, because this result may be probably mined by Bob legally on a \emph{candidate itemset}.
  \item $s\ll s_{\min}$. In this case, Bob has to enlarge $T$ to get $s$ without intolerable errors, for instance, $T>100/\sqrt{s}$. Since Bob does not know $s$ before computation, he has to adjust $T$ again and again \cite{BrassardHT1998} or directly sets a very large $T$. The comparison of costs for this case is given in Table \ref{table:ComparisonMrAdinD}.
\end{itemize}
\begin{table}
    \center
    \begin{tabular}{c|c|cc}
    $s$ & $N_R$ & C-cost & Q-cost\\
    \hline
    \multirow{2}{*}{$s\ll s_{\min}$} & 1 & $O(Nk)$  & $O((C_D+k)/\sqrt{s})$ \\
     & L & $O(Nk)$  & $O(L(C_D+k)/\sqrt{s})$ \\
    \end{tabular}
    \caption{Comparison of cost to check whether $d\in D$. In this table, $s$ is the frequency of transactions $d_j$, satisfying $d_j=d$. 
    $N_R$ is the number of different rules $d\in D$, which Bob wants to check. ``C-cost'' is the cost on classical database $D$ or $D'$, and ``Q-cost'' is the cost of multi-round attacks on quantum database $U_D$ or $U_{D'}$. $C_D$ is the cost to call $U_D$/$U_{D'}$ once. See Section \ref{sec:CombiningCQ} for $U_{D'}$.\label{table:ComparisonMrAdinD}}
\end{table}
In Table \ref{table:ComparisonMrAdinD}, we notice that usually it is cheaper to cheat in a quantum database $U_D$ if $L$ is small, since it is a search problem to check whether $d\in D$. So one method to overcome this weakness is to combine a classical protocol and a quantum one together: roughly speaking, Alice first modifies $D$ to another $D'$, and then runs the protocol on $U_{D'}$ (see Section \ref{sec:CombiningCQ} for more discussions about this point). After combining the classical and quantum protocols together, it is still faster for Bob to cheat on quantum database $U_{D'}$. But the information he gets on $U_{D'}$ is the same as that on $D'$. Then our quantum protocol is at least as good as a classical one for this function.

Before concluding this section, let us briefly consider Alice's strategy for multi-round attacks. In order to read small $s$, Bob requires large $T$. So Alice can set an upper bound for $T$. Then for a rule $d\in D$ with a low frequency, it can hardly be mined correctly with small $T$.

\section{Protecting Bob's Privacy}\label{Sec:PrivacyBob}
In Section \ref{Sec:PrivacyAliceAnalysis}, we showed from the Alice's side how our quantum protocol can protect privacy. In this section, we analyse Bob's privacy in terms of his functions $f$. For this purpose, it is certainly appropriate to assume that Bob itself is honest. Otherwise, if Bob is dishonest, he can protect his privacy by never sending $f$.

We first consider semi-honest Alice, and the analysis for honest Alice is similar. Semi-honest Alice follows the protocol, but she will do further computation based on measurement outcomes on the test states. In each test round, Alice may get the information about whether $f_i(\mu)=f_i(\nu)$ for some randomly generated $\mu,\nu$. Now we see how many pairs $(\mu,\nu)$ are required for this task. For association rule mining, there are totally at most $2^k$ functions (itemsets). For each pair, Alice compares $f_i(\mu)$ and $f_i(\nu)$, and gets one-bit information of $f_i(\mu)?=f_i(\nu)$. So, she has to build a $k$-level binary decision tree to include all possible $2^k$ leaves. Consequently, in general the number of test rounds is at least $k/2$ to recover one $f$ as there are two copies in each round. Since (1) there is at most only one test round in each loop $i$, and (2) the test round appears randomly, Alice can hardly get enough information to recover $f$. Moreover, as Bob adds noises into $f$ (see Section \ref{Sec:BobStrategyInner}), the information Alice gets may be wrong and thus useless. It is worth noting that some privacy leakage might happen in the last loop $i=T-1$. In Definition \ref{defi:BobStrategy}, no noise is added for $i=T-1$, which means $f_{T-1}=f$. So, if a test round appears, then Alice may know $f(\mu)?=f(\nu)$ for some randomly generated $(\mu,\nu)$. This leakage is not serious, since there are $2^{k-1}$ functions (itemsets) satisfying this one-bit property.

For dishonest Alice, the situation is different. Dishonest Alice may set a test at each loop and construct a special policy to choose test states to read information about $f$. The simple strategy in Definition \ref{defi:BobStrategy} is not sufficient to protect Bob's privacy. Fortunately Bob can protect his privacy from Alice's attacks by simply adding a second confusing qubit in Definition \ref{defi:BobStrategy}. Together with other methods, Bob can further improve his privacy level; see Appendix \ref{Apd:BobFurtherMethod} for details.

\begin{remark}
The privacy analysis for the case that Bob does not add noise and tests is postponed to Appendix \ref{Apd:PrivacyBob}.
\end{remark}

\section{Complexity Analysis}\label{Sec:Complexity}

The aim of this section is to analyse the complexity of our protocol. Actually, the cost of Algorithm \ref{alg:ProtocolInner} is easy to settle. Let us only consider association rule mining as an example. Suppose that the threshold of support is $s_{\min}$, and Bob totally run $M$ times of Algorithm \ref{alg:ProtocolInner} (since he wants to compute the supports of different itemsets or achieve a high accuracy by repetitions).
Then the \textit{computational complexity} is simply $$O(MT(C_D+k+n+t))=O(M(C_D+k+n+t)/\sqrt{s_{\min}}),$$ where $C_D$ is the cost of one call of Alice's database. So, if a \emph{quantum database} is available (e.g. as a quantum random access memory \cite{GiovannettiLM2008qram}), then $C_D=O(nk)$, and the total computational is $O(M(nk+t)/\sqrt{s_{\min}})$. Since $t=\log T$, and $T\leq \frac{\pi}{4}\sqrt{N}$ (meaning that the accuracy level is $\frac{1}{N}$), we have $t=O(n)$. So, the total computational complexity is $O(Mnk/\sqrt{s_{\min}})$. On the other hand, if the data is stored in a \emph{classical database}, then Alice has to use certain quantum gates to construct $O_D$, which costs  $O(N(n+k))$, and the total computational complexity is $O(MN(n+k)/\sqrt{s_{\min}})$.
The \textit{communication complexity} can be analysed similarly, and is $O(MT(t+n+k))=O(M(n+k)/\sqrt{s_{\min}})$. The results are illustrated in Table \ref{table:complexity}. Note that in many applications the communication cost may not be important and necessary. For a centralized database, Alice and Bob are at the same location, and we can imagine Alice as a preset database with an access to Bob. In this case, during communicating the problem of privacy including channel noise is not serious or even does not happen at all.
\begin{table}
    \center
    \begin{tabular}{c|cc}
        & T-cost & C-cost\\
        \hline
        Quantum database & $O(Mnk/\sqrt{s_{\min}})$ & $O(M(n+k)/\sqrt{s_{\min}})$\\
        Classical database & $O(MN(n+k)/\sqrt{s_{\min}})$ & $O(M(n+k)/\sqrt{s_{\min}})$\\
    \end{tabular}
    \caption{Cost of the entire quantum protocol, with data stored in a quantum database or a classical database. ``T-cost'' means computational complexity, and ``C-cost'' means communication complexity.\label{table:complexity}}
\end{table}

Now let us compare the complexity of our quantum protocol with that of a classical algorithm. Many different classical algorithms for the same task have been developed in the literature,  and each of them has a different cost and accuracy level. For those classical algorithms that require to input the whole database $D$ or $D'$ to achieve a better result, the computational and communication costs are both $O(Nk)$. Note that usually in practice $M\ll N$; for instance, $N=10^6$, and Bob might only care the most important hundreds of association rules with $M<10^3$. Then the costs of quantum protocol except the lower left entry of Table \ref{table:complexity} are better than those of classical algorithms, as $Nk>M(n+k)/\sqrt{s_{\min}}$.

\section{Discussions}\label{Sec:Discussions}

In this section, we point out several possibilities for further improvements of the protocol.

\subsection{Combining Classical and Quantum Protocols}\label{sec:CombiningCQ}
As mentioned before, combining our quantum protocol with a classical one is a way to further improve the privacy for Alice. In this strategy, there are totally two steps. The first step is to apply a classical approach on $D$ to get $D'$. Most of the classical approaches in the literature are suitable for this step; for example, randomly flipping elements in each transaction \cite{RizviH2002}, replacing elements partly \cite{EvfimievskiSAG2004}, swapping elements among different transactions \cite{EstivillB1999}.
The second step is to store $D'$ into a quantum database $O_{D'}$. Then our quantum protocol can be executed on $O_{D'}$.

The benefit of this method is obvious. Suppose that a classical approach changes $D$ to $D'$. Then in the classical case, Bob knows the entire $D'$. In the quantum case, however, Bob only knows a small part of $D'$, even if he is dishonest (see Section \ref{Sec:PrivacyAliceAnalysis}). So, this combination protects Alice's privacy much better
than solely using a classical approach.

The disadvantage is that additional error may be introduced. A combination of our quantum protocol and a classical one has two places to generate errors: one is from randomization in the classical protocol, the other is from the quantum counting. Thus, the total error may be larger than that of the classical algorithm.

\subsection{Decision Tree Learning}\label{sec:DecisionTree}
Our protocol was presented mainly for association rule mining, but Algorithm \ref{alg:ProtocolInner} can be directly used to mine decision trees. Consider the basic algorithm for decision tree mining proposed in \cite{Quinlan1986}. Here, we show how it can be combined with Algorithm \ref{alg:ProtocolInner} so that privacy can be protected.
\begin{example}
    Suppose that Alice holds a database 
    $$D=\<(d_0,g_0), (d_1,g_1), \cdots, (d_{N-1},g_{N-1})\>,$$ where $d_j\in\{0,1\}^k$, $g_j=f(d_j)\in\{0,1\}$ for a function $f$. Suppose that $K$ is the set of all first $k$ attributes. Bob wants to build a decision tree, with one attribute in $K$ at each node, to decide $f(d)$  for any input $d$, based on the database $D$. The algorithm is shown below:
    \begin{enumerate}
      \item Set $L=0$, and the root $r$ is set to be an empty node.
      \item For each empty node on level $L$, computes its corresponding attribute:
      \begin{itemize}
        \item Suppose $F$ is the set of attributes corresponding to the ancestors of this node.
        \item Denote $A=\{y:y\in F\}$. Bob computes the support $s_0$ of $(A,0)$ and $s_1$ of $(A,1)$ by Algorithm \ref{alg:ProtocolInner}. If $H(A) = -\sum_i s_i\log s_i$ is smaller than a preset threshold $H_{\min}$, the node is set to be value 0 (if $s_0>s_1$) or 1 (if $s_0<s_1$) for $f$. No child is generated. Return.
        \item For each attribute  $x\in K\setminus F$, denote $A_x=\{x\}\cup\{y:y\in F\}$. Alice computes the support $s_0(x)$ of $(A_x,0)$ and $s_1(x)$ of $(A_x,1)$ by Algorithm \ref{alg:ProtocolInner}. Then she computes the entropy $H(x) = -\sum_i s_i(x)\log s_i(x)$.
        \item Bob chooses the attribute $y$ which maximizes $H(y) = \max H(x)$. The corresponding attribute of this node is set to be $y$, and two children are generated. Each is for value 0 or 1 of $y$.
      \end{itemize}
      \item If no child is generated in this level $L$, terminates. Otherwise, $L:=L+1$ and goes to Step 2.
    \end{enumerate}
\end{example}

The privacy and complexity analyses for the above example are similar to what we did for association rule mining in the previous sections.

\subsection{Dishonesties of Both Alice and Bob}
In this paper, we presented a method for Alice to deal with dishonest Bob, and also a method for Bob to deal with dishonest Alice. In particular, we showed:
\begin{itemize}
  \item If Alice and Bob are both honest, our protocol computes the final results and preserves privacy for both parties.
  \item If Alice is honest and Bob is dishonest, our protocol can (1) compute the final results, (2) detect Bob's attack to protect Alice's privacy, and (3) preserve Bob's privacy.
  \item If Alice is dishonest and Bob is honest, our protocol can (1) compute the final results, and (2) preserve privacy for both Alice and Bob {  (by adding a second confusing qubit)}.
\end{itemize}
Then a question naturally arises: what happens when both Alice and Bob are dishonest? In this case, there is no definite conclusion. It depends on what actions are taken. For example, suppose: \begin{itemize}
  \item Alice acts honestly if $i$ is odd. If $i$ is even, she stores the computational state aside for the next loop, and tries to read $f(\mu)$ by sending Bob a state like $\frac{1}{\sqrt{2}}|c\>|0\>^{\otimes n-1}(|0\>|\mu\>+|1\>|\vec{1}\>)$.
  \item Bob acts honestly if $i$ is even. If $i$ is odd, he performs measurements to read $d_j$, and sends to the post-measurement state back to Alice.
\end{itemize}
Then in each loop $i$, either Alice or Bob cheats, and the computation cannot be accomplished. But if Bob's actions are switched, then both Alice and Bob act honestly when $i$ is odd. Furthermore, if Alice stores the computational states properly when $i$ is even, then the computation can be accomplished (but with a larger error).

\subsection{Compatibility with Other Quantum Algorithms}\label{Sec:PureDatabase}
Note that in the protocol, Alice's part is compatible with other quantum algorithms in the sense that if Bob wants to run a quantum algorithm on Alice's database other than quantum counting, he only needs to modify his part of protocol. For example, suppose Bob wants to find whether a given transaction $d$ is in the database by running a quantum walk on a hyper cube \cite{ShenviKW2003}, where each node corresponds to an address $j$ and a transaction $d_j$ (or $d_{j\oplus y}$). In a quantum walk-based search, the operator $G$ is replaced by some other operators. What we need to do are the following modifications on Bob's part:
\begin{itemize}
  \item Change the initial state of the control qubits. In Algorithm \ref{alg:ProtocolInner}, Alice does not check the initial state on control qubits. So for the searching problem, Bob can simply set all control qubits to be $|1\>$.
  \item Remove Step \ref{line:BFourier} in Algorithm \ref{alg:ProtocolInner} and change his actions at Step \ref{line:L2} in procedure \ref{alg:ProtocolLoop}.
\end{itemize}
It is easy to see that privacy of both parties is preserved in the same way as our original protocol.

\subsection{Parameter $p$ in Algorithm \ref{alg:ProtocolInner}}
The parameter $p$ in Algorithm \ref{alg:ProtocolInner} indicates how frequently tests are employed. Obviously, for different models, the best choice of $p$ varies. Note that $p$ only matters in (1) detecting Bob's one-round attacks (dishonest actions), and (2) disclosing Bob's privacy $f$ by comparing $f(\mu)$ and $f(\nu)$. So, $p=0$ or $p\ra 0$ is preferred for honest or semi-honest Bob. If Bob is dishonest, the situation becomes quite different:
\begin{itemize}
  \item Honest Alice: A big $p$ is preferred, as it protects Alice's privacy better than small $p$.
  \item Semi-honest Alice: Either a big or small $p$ is not the best choice, since one party's privacy is likely to be disclosed in both cases. So, a medium $p$ is the most suitable choice.
  \item Dishonest Alice: No best choice exists, because the protocol may not work if Alice always cheats.
\end{itemize}

The preferred choices of $p$ are summarised in Table \ref{table:ModelP}.
\begin{table}
    \center
    \begin{tabular}{c|ccc}
      \diagbox{Alice}{Bob} & Honest & Semi-honest & Dishonest\\
      \hline
      Honest & $p\ra 0$ & $p\ra 0$ & Big $p$\\
      Semi-honest & $p\ra 0$ & $p\ra 0$ & Medium $p$\\
      Dishonest &  $p\ra 0$ & $p\ra 0$ & --\\
    \end{tabular}
    \caption{Preferred choice of $p$ in Algorithm \ref{alg:ProtocolInner} for different models. \label{table:ModelP}}
\end{table} However, Alice and/or Bob cannot know which situation they are facing. So, Table \ref{table:ModelP} is helpless in practice. Generally speaking, Alice prefers a big $p$. But if $p$ is too big, Bob's privacy will be disclosed when Alice is not honest. In Section \ref{Sec:PrivacyAliceAnalysis}, it was shown that if $N$ is big enough, the ratio of information (transactions) disclosed is nearly 0 for any $p\in (0,1)$. So, in practice, $p=0.05$ may work well. Indeed, the most important implication of parameter $p\in (0,1)$ is not to detect Bob's privacy but to tell Bob that once he cheats, he may be caught. This fact may force Bob to be honest.

\newpage
\appendix

\section{Procedure \ref{alg:TestACDG}} \label{Apd:TestACD2ACDG}

In this Appendix, we present the detailed description of procedure \ref{alg:TestACDG} that was only very briefly discussed in Section \ref{main protocol}.
\begin{procedure}
    \SetKwData{Left}{left}\SetKwData{This}{this}\SetKwData{Up}{up}
    \SetKwFunction{Union}{Union}\SetKwFunction{FindCompress}{FindCompress}
    \SetKwInOut{Input}{input}\SetKwInOut{Output}{output}
    \Output{ ``Dishonesty detected'', if Bob's dishonesty is detected.}
    \Begin{
    Alice generates $\mu<\nu\in\{0,1\}^k$, $c\in\{0,1\}^t$, $m\in\{0,1,\cdots,n-1\}$, $x\in \{0,1\}^n$, and $b\in\{0,1\}$ uniformly at random\;
    Alice prepares $|\Phi\> = |c\>_{q_c}\otimes U_t(m,x,b)|0\>^{\otimes (n+k)}_{q_{a}, q_{d}}$ {  on new control qubits $q_c$, address qubits $q_a$, and data qubits $q_d$}\;\label{line:T31}
    Bob adds a qubit $q_{g}$ initialized to be $|0\>$ to the end of these qubits, and Bob applies $U'_{f}$ on $q_{a}$, $q_{d}$, and $q_g$\;
    Alice applies $U_t(m,x,b)^\dag$ on $q_{a}$ and $q_{d}$, obtaining the state $$|\Phi_1\> =  |c\>_{q_c}\otimes U_t(m,x,b)^\dag\otimes I_g (U'_{f} (U_t(m,x,b)|0\>^{\otimes (n+k)}\otimes |0\>));$$\label{line:T39}\\
    \vspace{-1em}
    Alice and Bob repeat Step \ref{line:T31} to Step \ref{line:T39} to get $|\Phi_2\>$\;
    Alice employs a quantum controlled swap test to test whether $|\Phi_1\>$ and $|\Phi_2\>$ are the same. If not, Alice terminates the entire protocol\;\label{line:T3ControlledSwapTest}
    Alice measures all address and data qubits except the first address qubit of $|\Phi_1\>$, and $|\Phi_2\>$, according to the basis $\{|0\>,|1\>\}$. Let the outcomes be $v$ and $w$, respectively, both in $\{0,1\}^{n+k-1}$\;
    {  If  $v_j=1$, or $w_j=1$ for any $j$, \Return{``Dishonesty detected''}\;\label{line:T3Terminating}}
    }
    \caption{TestBob2()}\label{alg:TestACDG}
\end{procedure}

\subsection{Controlled Swap Test}\label{Sec:ControlledSwapTest}
Controlled swap tests are employed at Step \ref{line:T3ControlledSwapTest} in procedure \ref{alg:TestACDG} to check whether Bob performs measurements on $q_g$ or $q_d$. In this subsection, we briefly describe these tests. For details, we refer to \cite{ControlledSwapTest}.

A quantum swap gate consists of three CNOT gates, and swaps the states of two qubits:
\begin{align*}
    SWAP: \sum_{i,j\in\{0,1\}}\alpha_{i,j}|i\>|j\> \ra\sum_{i,j}\alpha_{i,j}|i\>|j\oplus i\>\ra\sum_{i,j}\alpha_{i,j}|j\>|j\oplus i\>\ra \sum_{i,j}\alpha_{i,j}|j\>|i\>.
\end{align*}
Then a controlled swap test on two $n$-qubit states $|\psi\>$ and $|\phi\>$ works as follows:
\begin{enumerate}
  \item Add an ancilla qubit in state $|+\>$ before $|\psi\>$ and $|\phi\>$, and get
    \begin{equation*}
        |\Psi_1\> = |+\>|\psi\>|\phi\>.
    \end{equation*}
  \item Apply a controlled swap operator $U_{CS}$ on $|\Psi_1\>$, where the first qubit is the control qubit, and the other qubits are the target:
    \begin{equation*}
        |\Psi_2\>=U_{CS}|\Psi_1\> = \frac{1}{\sqrt{2}}(|0\>|\psi\>|\phi\>+|1\>|\phi\>|\psi\>),
    \end{equation*}
    where $|\psi\>$ and $|\phi\>$ are swapped if the control qubit is in state $|1\>$.
  \item Apply a Hadamard gate on the control qubit, and get
    \begin{align*}
        |\Psi_3\> &= \frac{1}{\sqrt{2}}(|+\>|\psi\>|\phi\>+|-\>|\phi\>|\psi\>)\\
        &= \frac{1}{{2}}(|0\>(|\psi\>|\phi\>+|\phi\>|\psi\>)+|1\>(|\psi\>|\phi\>-|\phi\>|\psi\>)).
    \end{align*}
    If $|\psi\>=|\phi\>$, we have $|\Psi_3\> = |0\>|\psi\>|\phi\> $.
  \item Measure the control qubit in basis $\{|0\>,|1\>\}$. If outcome is 1, $|\psi\>\neq|\phi\>$ is detected.
\end{enumerate}
At Step 4, the probability to get outcome 1 is
\begin{align*}
    \Pr(1) = \frac{1}{{4}}\||\psi\>|\phi\>-|\phi\>|\psi\>\|^2      = \frac{1}{{4}}(2-2\<\phi|\psi\>\<\psi|\phi\>)
    = \frac{1-|\<\phi|\psi\>|^2}{2}.
\end{align*}
Omitting the control qubit, the post-measurement states are
\begin{equation*}
    \begin{cases}
        |\varphi_0\> = \frac{1}{{2}}(|\psi\>|\phi\>+|\phi\>|\psi\>) & \text{if outcome is}\ 0\\
        |\varphi_1\> = \frac{1}{{2}}(|\psi\>|\phi\>-|\phi\>|\psi\>) & \text{if outcome is}\ 1.
    \end{cases}
\end{equation*}
In conclusion, if $|\psi\>=|\phi\>$, the outcome is always 0, and the state $|\psi\>|\phi\>$ remains unchanged. Otherwise, there is a probability $\frac{1-|\<\phi|\psi\>|^2}{2}$ to have outcome 1.

\section{Proofs of Lemmas and Theorems}

In this Appendix, we provide the proofs of the lemmas and theorems presented in the main text.
\subsection{Proof of Lemma \ref{lem:Attack1C}}
    We first specify the situation:
    \begin{enumerate}
      \item Alice sends $|c\>|\psi_{m,x,b}(\mu,\nu)\>$ to Bob.
      \item After Bob's actions, he sends a state to Alice; see Table \ref{Table:Attack1C}.
      \item Alice checks the test state.
    \end{enumerate}
    \begin{table}
        \center
        \begin{tabular}{l|c|c}
            Test & State if honest & State in this attack\\\hline
            \ref{alg:TestACD1} &  $|c\>|\psi_{m,x',b}(\mu,\nu)\>$ & $|\varphi\>|i\>|d\>$ \\
        \end{tabular}
        \caption{The state Bob sends to Alice. In this table $x'$ depends on whether $f(\mu)=f(\nu)$.\label{Table:Attack1C}}
    \end{table}

    Now we prove that it can be detected by the final measurement (Step \ref{line:TTerminating}).

    \textbf{Case 1} $d=\mu$: Let $|\psi'\>$ denote the state that Bob sends back to Alice. Then $|\psi'\>=|\varphi\>|a\>|\mu\>$. Alice will transfer it to
    \begin{equation*}
        |\phi'\> = V(\mu,\nu)X_0^{b}Z(x)U_{SWAP(0,m)}|\psi'\>,
    \end{equation*}
    and then
    \begin{equation*}
        |\Phi_1\> =  I_c\otimes W\otimes I_d |\phi'\>.
    \end{equation*}

    Now we see what $|\Phi_1\>$ is. The initial state is $|\psi'\>=|\varphi\>|a\>|\mu\>$. Then
    \begin{enumerate}
      \item $U_{SWAP(0,m)}$ maps $a$ to $a'\in\{0,1\}^n$;
      \item $Z(x)$ only possibly adds a global phase $-1$;
      \item $X^b_0$ maps $a'$ to $a''$;
      \item $V(\mu,\nu)$ only works on the data qubits, and then $|\phi'\>$ can be $|a''\>|\vec{0}\>$ or $|a''\>|\mu\oplus\nu\>$, after omitting the global phase.
    \end{enumerate}
    Consequently, $|\phi_1\>$ may be $|w\>|*\>$, where $w\in\{+,-\}^n$. On each address qubit, the measurement outcome will be 0 with probability $\frac{1}{2}$. Therefore, this attack passes the final measurement (Step \ref{line:TTerminating}) with probability at most $\frac{1}{2^{n-1}}=\frac{2}{N}$.

    Since the test consists of two copies, the attack can be further detected. Suppose that Bob sends $|\varphi\>|a\>|\mu\>$ in the first round. By the above discussion, this means that $v_0$ has a fifty-fifty chance to be 0 or 1. For the other round, we assume that $w_0$ comes out as 0 with probability $p$, and 1 with probability $1-p$. Then Bob has probability $\frac{1}{2}(1-p)+\frac{1}{2}p = \frac{1}{2}$ to fail. Thus, the total probability to pass the test is at most $\frac{1}{N}$.

    \textbf{Case 2}. $d=\nu$: Similar.

    \textbf{Case 3}. $d\neq\mu$ and $d\neq\nu$: The probability to pass the final measurement may be lower. This is because there must be at least one $|1\>$ occurring in the data qubits of $|\Phi_1\>$ at last. So, it will be detected with probability 1.

\subsection{Proof of Lemma \ref{lem:Attack1M}}
    In the scenario of this lemma, Bob thinks that he receives $|\varphi\>|a\>|d_{a\oplus y}\>$ from Alice. But indeed, he receives $|c\>|\psi_{m,x,b}(\mu,\nu)\>$. Then the measurements are taken on $|c\>|\psi_{m,x,b}(\mu,\nu)\>$.

    If the measurements are performed jointly on $q_a$ and $q_d$, then the post-measurement state is $|c\>|j\>|\mu\>$ or $|c\>|j\>|\nu\>$, and it can be reduced to Lemma \ref{lem:Attack1C} as Bob will send this state back to Alice. So, it remains to analyse the case when measurements are on only $q_d$.

    Since measurements are only performed on $q_d$, 
    all the measurement operators are local operators on $q_d$ and have the form $M=I^{\otimes n} \otimes |w\>\<w|$, where $w\in\{0,1\}^k$. Thus, each $M$ commutes with $X_0$, $Z(x)$ and $U_{SWAP(0,m)}$, as the latter three operators are local operators on $q_a$. Therefore, the post-measurement state of $|c\>|\psi_{m,x,b}\>$ can be written as
    \begin{equation*}
        |\psi'\>=\begin{cases}
            |c\>U_{SWAP(0,m)}Z(x)X_0^b |0\>|+\>^{\otimes n-1}|\mu\>, & \text{the outcome is~} \mu\\
            |c\>U_{SWAP(0,m)}Z(x)X_0^b |1\>|+\>^{\otimes n-1}|\nu\>, & \text{the outcomes is~} \nu
        \end{cases}.
    \end{equation*}
    Bob sends it back to Alice, and then Alice has
    \begin{align*}
        |\phi'\>=\begin{cases}
           |c\> V(\mu,\nu) |0\>|+\>^{\otimes n-1}|\mu\>=|0\>|+\>^{\otimes n-1}|\vec{0}\>, & \text{the case~} \mu\\
           |c\>V(\mu,\nu) |1\>|+\>^{\otimes n-1}|\nu\>=|1\>|+\>^{\otimes n-1}|\vec{0}\> , & \text{the case~} \nu
        \end{cases},
    \end{align*}
    Furthermore, it holds that
    \begin{equation*}
        |\Phi_1\>=\begin{cases}
            |c\>|+\>|0\>^{\otimes n+k-1}, & \text{the case~} \mu\\
            |c\>|-\>|0\>^{\otimes n+k-1}, & \text{the case~} \nu
        \end{cases}.
    \end{equation*}
    In both cases, the first element $v_0$ of the final outcome result $v$ has a fifty-fifty probability to become 0 or 1. This means that this attack can be detected by \ref{alg:TestACD1} with probability $\frac{1}{2}$ by the condition $v_0\neq w_0$.

\subsection{Proof of Lemma \ref{lem:Attack3}}
    In this proof, we omit $|c\>$, since it indeed does not change through the test. Moreover, we use $|\Psi'\>$, $|\Phi'\>$ and  $|\Phi'_1\>$ on $q_a, q_d, q_g$ to replace $|\psi'\>$, $|\phi'\>$ and  $|\Phi_1\>$, respectively.

    (1) The general case. By the assumption of this lemma, there exist some $i,i'\in\{0,1\}^n$, and $d,d'\{0,1\}^k$ such that $|\lambda_{i,d}\>\neq|\lambda_{i',d'}\>$. We can further assume $i\neq i'$ and $d\neq d'$. (Otherwise,  if $i=i'$ or $d=d'$, we choose $i''\not\in\{i,i'\}$ and $d''\not\in\{d,d''\}$, and then $|\lambda_{i'',d''}\>$ must be different to one of the original two.)

    Suppose $d<d'$. By the construction of test states, we find $\mu=d$, and $\nu=d'$ with probability $\frac{2}{2^k(2^k-1)}$. As $i\neq i'$, there exists $m$ such that $i_m\neq i'_m$. Then for $b=i_m$ and any $x$, we can find both item $|i\>|d\>$ and $|i'\>|d'\>$ in $|\psi_{m,x,i_m}(\mu,\nu)\>$. The total probability to generate this state  is $$\frac{2}{2^k(2^k-1)}\frac{1}{n}\frac{1}{2}=\frac{1}{n2^k(2^k-1)}.$$

    Now $|\psi_{m,x,i_m}(\mu,\nu)\>$ is entangled to be $|\Psi_{m,x,i_m}(\mu,\nu)\>$. Alice transforms it to
    \begin{align*}
        |\Phi'_1(\mu,\nu)\>=(V(\mu,\nu)X_0^{i_m}Z(x)U_{SWAP(0,m)})\otimes I_g|\Psi'_{m,x,i_m}(\mu,\nu)\>.
    \end{align*}
    Since these operators are unitary and only on the address and data qubits, the address and data qubits are still entangled to $q_g$. Then no matter whether measurement operator $M_v=|v\>\<v|$ or  $M_{v'}=|v'\>\<v'|$ is used with $|v\>=|0\>|0\>^{\otimes n+k-1}$ and $|v'\>= |1\>|0\>^{\otimes n+k-1}$, it holds
    \begin{equation*}
    \begin{cases}
        p_0=\|M_v\otimes I_g |\Phi'_1(\mu,\nu)\>\|^2<1,\\
        p_1=\|M_{v'}\otimes I_g |\Phi'_1(\mu,\nu)\>\|^2<1.
    \end{cases}
    \end{equation*}
    Thus, if $p_0+p_1<1$, it will be detected by $v_j=1$ with $j>0$. Otherwise it will be still detected by the condition $v_0\neq w_0$ with a positive probability as $p_0<1$ and $p_1<1$.

    (2) The case $E|i\>|d\>|{0}\> = |i\>|d\>|i\>|d\>$.  We have:
    \begin{align*}
        |\Phi'(\mu,\nu)\>=&(V(\mu,\nu)X_0^{i_m}Z(x)U_{SWAP(0,m)})\otimes I_q\frac{1}{\sqrt{N}}\sum_{i}\beta_i|i\>|\tau_i\>|i\>|\tau_i\>\\
        =& \frac{1}{\sqrt{N}}\sum_{i}\beta_i(V(\mu,\nu)X_0^{i_m}Z(x)U_{SWAP(0,m)}|i\>|\tau_i\>)|i\>|\tau_i\>\\
        =& \frac{1}{\sqrt{N}}\sum_{i}\beta_i|i'\>|\tau'_i\>|i\>|\tau_i\>,
    \end{align*}
    where $\beta_i\in\{1,-1\}$, $\tau_i\in\{\mu,\nu\}$, $V(\mu,\nu)X_0^{i_m}Z(x)U_{SWAP(0,m)}$ is a bijection mapping $i$ to $i'$, and $\tau'_i\in\{0,\mu\oplus\nu\}$. Then
    \begin{align*}
        &|\Phi'_1(\mu,\nu)\>= \frac{1}{\sqrt{N}}\sum_{i}\beta_i(W|i'\>)|\tau'_i\>|i\>|\tau_i\>.
    \end{align*}
    Since each item has a $|i\>$ which belongs to $q_g$ and is orthogonal to each other, probabilities $p_0$ and $p_1$ can be calculated:
    \begin{align*}
        p_0=&\|M_v\otimes I_g |\Phi'_1(\mu,\nu)\>\|^2\\ =& \|\frac{1}{\sqrt{N}}\sum_{i}\beta_i(\<00\cdots 000|W|i'\>)\<00\cdots 00|\tau'_i\>|i\>|\tau_i\>\|^2\\
        =&\frac{1}{N}\sum_{i}\|(\<00\cdots 000|W|i'\>)\<00\cdots 00|\tau'_i\>\|^2\\
        =&\frac{1}{N^2}\sum_{i}\|\<00\cdots 00|\tau'_i\>\|^2\leq \frac{1}{N},
    \end{align*}
   Similarly, we obtain:
    \begin{align*}
        p_1=&\frac{1}{N}\sum_{i}\|(\<10\cdots 000|W|i'\>)\<00\cdots 00|\tau'_i\>\|^2\\
        =&\frac{1}{N^2}\sum_{i}\|\<00\cdots 00|\tau'_i\>\|^2\leq \frac{1}{N}.
    \end{align*}
    Therefore the probability to fail is
    \begin{align*}
        \Pr(deteceted)
        =&1-p_0-p_1+(p_0+p_1)\Pr(v_0\neq w_0)\geq 1-\frac{1}{N}.
    \end{align*}

    (3) For the case $E|i\>|d\>|{0}\> = |i\>|d\>|d\>$, as in the proof of Lemma \ref{lem:Attack1M}, we have
    $$|\Phi'_1(\mu,\nu)\>=\frac{1}{\sqrt{2}}(|0\>|0\>^{\otimes n+k-1}|\mu\>+|1\>|0\>^{\otimes n+k-1}|\nu\>).$$
    Then we have $p_0=p_1=\frac{1}{2}$. It will be detected by $v_0\neq w_0$ with probability $\frac{1}{2}$.

\subsection{Proof of Theorem \ref{thm:strongPassMeasurement}}
We check what will happen, if the state is entangled to $q_g$. Assume the test states are $|\psi_{m,x,b}(\mu,\nu)\>$. For simplicity, we denote it by $|\psi\>$. Suppose Bob employ $\E_1$ and $\E_2$ on the two test states with new ancilla qubits (initialed to be $|\theta_1\>$ and $|\theta_2\>$ respectively). Then we have
\begin{equation*}
    \begin{cases}
        \E_1(\psi\otimes \theta_1) =\rho= \sum_u \lambda_u |\varphi_u\>\<\varphi_u|,\\
        \E_2(\psi\otimes \theta_2) =\sigma= \sum_v \chi_v |\omega_v\>\<\omega_v|,
    \end{cases}
\end{equation*}
where $\lambda_u,\chi_v\in[0,1]$, and $\{|\varphi_u\>\}$, $\{|\omega_v\>\}$ are orthonormal bases. Since density operators can be seen as probabilistic distributions over pure states, the following two facts are equivalent:
\begin{itemize}
  \item $\rho$ and $\sigma$ pass the test with probability 1.
  \item For any $u,v$, $|\varphi_u\>$ and $|\omega_v\>$ pass the test with probability 1.
\end{itemize}
Suppose $|\varphi\>$ stands for any $|\varphi_u\>$ and $|\omega\>$ stands for any $|\omega_v\>$, and
\begin{equation*}
    \begin{cases}
        |\varphi\> = \sum_j \alpha_j|\xi_j\>|j\>,\\
        |\omega\> = \sum_j \beta_j|\gamma_j\>|j\>,
    \end{cases}
\end{equation*}
where $\{|j\>\}$ is an orthonormal basis of the Hilbert space $\hs_g$ on $q_g$, and $|\xi_j\>$, $|\gamma_j\>$ are normalized states ($\sum_j |\alpha_j|^2 = \sum_j |\beta_j|^2=1$). Suppose the same recovery operator for these two states in the test is $R$. Then after recovery, the states become
\begin{equation*}
    \begin{cases}
        |\bar{\varphi}\> = \sum_j \alpha_j(R|\xi_j\>)|j\>,\\
        |\bar{\omega}\> = \sum_j \beta_j(R|\gamma_j\>)|j\>.
    \end{cases}
\end{equation*}
Note that Bob holds the second part of $|\varphi'\>$ and $|\omega'\>$, Alice performs measurements only on the first part. Then in order to pass Step \ref{line:TTerminating}, the measurement outcomes are always fixed, and the above states must be
\begin{equation*}
    \begin{cases}
        |\bar{\varphi}\>= |\tau\>|0\>^{\otimes n+k-1}|\varphi'\>,\\
        |\bar{\omega}\>= |\tau\>|0\>^{\otimes n+k-1}|\omega'\>,
    \end{cases}
\end{equation*}
where $\tau\in\{0,1\}$, and $|\varphi'\>, |\omega'\>$ are states of $q_g$. Since $R$ is a unitary operator, we have
\begin{equation*}
    \begin{cases}
        |\varphi\> = R^\dag\otimes I_g |\bar{\varphi}\> = |\xi\>\otimes|\varphi'\>,\\
        |\omega\> = R^\dag\otimes I_g |\bar{\omega}\> = |\xi\>\otimes|\omega'\>.
    \end{cases}
\end{equation*}
Moreover, since $|\varphi\>$ stands for any $|\varphi_u\>$ and $|\omega\>$ stands for any $|\omega_v\>$, $|\xi\>$ is independent of $u,v$. Then we have
\begin{equation*}
    \begin{cases}
        \E_1(\psi\otimes \theta_1) =\rho= |\xi\>\<\xi|\otimes \rho',\\
        \E_2(\psi\otimes \theta_2) =\sigma= |\xi\>\<\xi|\otimes \sigma'.
    \end{cases}
\end{equation*}
Since all possible test states $|\psi\>$ can be used to construct several orthonormal bases (see Lemma \ref{lem:TestStateBasis}), $\E_1$ can be written as
\begin{equation*}
    \E_1 = U\circ U^\dag \otimes \E_g,
\end{equation*}
where $U$ is a unitary operator on $q_a$ and $q_d$, and $|\xi\> = U|\psi\>$. The conclusion about $\E_2$ can be proved similarly.

Now we prove that $\rho'$ ($\sigma'$) is independent on $\xi$. By the above explicit form of $\E_1$, it seems obvious intuitively. But here we give a strict proof. Suppose $|\psi_1\>$ and $|\psi_2\>$ are two test states for short. Suppose $\E(\psi_i\otimes \theta) =  U|\psi_i\>\<\psi_i| U^\dag\otimes \rho'_i$. Since the fidelity $\digamma(\psi_i\otimes \theta, \psi_i\otimes \theta)$ will increase after super-operator $\E$, we have
\begin{align*}
    \digamma(\psi_1\otimes \theta, \psi_2\otimes \theta) &\leq \digamma(\E(\psi_1\otimes \theta), \E(\psi_2\otimes \theta))\\
    &= \digamma(U|\psi_1\>\<\psi_1| U^\dag\otimes \rho'_1, U|\psi_2\>\<\psi_2| U^\dag\otimes \rho'_2)\\
    & = \digamma(U|\psi_1\>\<\psi_1| U^\dag,U|\psi_2\>\<\psi_2| U^\dag)\digamma(\rho'_1,\rho'_2)\\&\leq \digamma(U|\psi_1\>\<\psi_1| U^\dag,U|\psi_2\>\<\psi_2| U^\dag)\\
    & =\digamma(|\psi_1\>\<\psi_1|,|\psi_2\>\<\psi_2|).
\end{align*}
On the other hand, $$\digamma(\psi_1\otimes \theta, \psi_2\otimes \theta)=\digamma(\psi_1, \psi_2)\digamma(\theta,\theta)=\digamma(\psi_1, \psi_2).$$
Therefore, $\digamma(\rho'_1,\rho'_2)=1$, i.e., $\rho'_1=\rho'_2$. This means $\rho'$ is independent on $|\psi\>$.

\subsection{Proof of Corollary \ref{cor:AlwaysPassTest}}
Suppose Bob employs certain operators and measurements in one test round. Then all of Bob's actions are equal to the measurement with measurement operators $E_j$. Note these measurement operators  form the super-operator $\E=\sum E_j\circ E_j^\dag$ (performing a measurement without reading the outcomes is equivalent to performing a super-operator). By Theorem \ref{thm:strongPassMeasurement}, each measurement operator $E_j$ has the form $E_j= U\otimes M_j$. So, on $q_a$ and $q_d$, no measurements are performed, which means that Bob cannot read any information from $q_a$ and $q_d$.

\subsection{Proof of Lemma \ref{lem:RemoveingTest2}}
    In order to read information about specific $d$ from $D$, Bob needs to perform measurements on a state like $\sum_{c,j}\alpha_{c,j}|c\>|j\>|d_{j\oplus y}\>,$ instead of a state like $\sum_{c,j}\alpha_{c,j}|c\>|j\>|\vec{0}\>.$ But if Bob is honest or he does not construct a new state to send to Alice at some step, the state he gets at Step \ref{line:L2} is always the latter. So Bob has to first send a state to Alice at some step, and then perform measurements.

There are three different ways to attacks at Step \ref{line:L2} of Algorithm \ref{alg:ProtocolInner} here. The first two ways correspond to direct attacks (reading information) at Step \ref{line:L2}. Bob can first send a cheating state (for instance, $|c\>|j\>|\vec{0}\>$) back to Alice at Step \ref{line:L1} or Step \ref{line:L2}, and then perform measurements on the received state (for instance, $|c\>|j\>|d_{j\oplus y}\>$) at Step \ref{line:L2}. The last way is to do some preparation at Step \ref{line:L2} for a future attack. In this way, Bob sends a cheating state back to Alice at Step \ref{line:L2}, and then performs measurements on the received state at Step \ref{line:L1} (Note that sending a cheating state and performing measurements both at Step \ref{line:L1} is not related to the attacks at Step \ref{line:L2}.)
    Obviously, each of the above attacks can be divided into two parts: (1) sending state; and (2) performing measurements. All of the possible cases are listed in Table \ref{Table:RemovingT2}. We analyse these cases one by one.
    \begin{table}
    \center
        \begin{tabular}{c|cc}
            & Sending a cheating state & Performing measurements\\\hline
            1 & Step \ref{line:L1} & Step \ref{line:L2}\\
            2 & Step \ref{line:L2} & Step \ref{line:L2}\\
            3 & Step \ref{line:L2} & Step \ref{line:L1}
        \end{tabular}
        \caption{Three different ways related to one-round attacks at Step \ref{line:L2}.}\label{Table:RemovingT2}
    \end{table}

    Case 1. Bob's measurements at Step \ref{line:L2} will not be detected. But the part of sending a cheating state at Step \ref{line:L1} will be detected. For instance, if Bob sends state $|c\>|j\>|\vec{0}\>$ at Step \ref{line:L1}, it will be detected by procedure \ref{alg:TestACD1} with probability at least 0.5 (See Lemma \ref{lem:Attack1C}). If Bob sends any other state, it can be also detected by the results in Section \ref{Sec:PrivacyAliceAnalysis}.

    Case 2. This case cannot be detected, if there is no test corresponding to Step \ref{line:L2}. But Bob cannot read any information about some specific $d\in D$ from this kind of attacks. Indeed, from procedure \ref{alg:ProtocolLoop}, one sees that there are an even number of calls of database $U_D(y)$ on the cheating state from sending it to receiving it. So, if the cheating state is $\sum\beta_{c,j}|c\>|j\>|\vec{0}\>$, it becomes $\sum\beta'_{c,j}|c\>|j\>|\vec{0}\>$ when Bob receives it again. There is still no information about specific $d\in D$ on the data qubits. Therefore, this attack does not work.

    Case 3. Similar to Case 1, this attack will be detected at Step \ref{line:L1} after measurements.

\subsection{Proof of Lemma \ref{lem:TestStateBasis}}
    We first prove the following technical lemma.
    \begin{lemma}\label{lem:weakBasis}
        Suppose $B_m(\mu,\nu) = \{|\psi_{m,x,b}(\mu,\nu)\>:\forall x,b\}$. Then $B_m$ is an orthogonal basis of Hilbert space spanned by $\{|i\>|\mu\>,|i\>|\nu\>:i\in\{0,1\}^n\}$.
    \end{lemma}
    \begin{proof}
    It is sufficient to prove that $B_{n-1}(\mu,\nu)$ is an orthogonal basis. Observe that any $|\psi_{n-1,x,b}(\mu,\nu)\>$ can be rewritten as $$|\psi_{n-1,x,b}(\mu,\nu)\>= |\phi_{x_1,\cdots,x_{n-1}}\>\otimes |\beta_{x_0,b}\>,$$
    where $|\phi_{x_1,\cdots,x_{n-1}}\> = Z^{x_{n-1}}\otimes Z^{x_1}\otimes Z^{x_2}\otimes \cdots Z^{x_{n-2}}|++\cdots +\>$, and $$|\beta_{x_0,b}\>=Z^{x_0}X^b\frac{1}{\sqrt{2}}(|0\>|\mu\>+|1\>|\nu\>).$$ Since $Z|+\>=|-\>$, all $|\phi_{x_1,\cdots,x_{n-1}}\>$ form an orthogonal basis of the first $n-1$ qubits. Meanwhile, four different states of $|\beta_{x_0,b}\>$ form an orthogonal basis of the other part. Thus,  $B_{n-1}(\mu,\nu)$ is an orthogonal basis of Hilbert space spanned by $\{|i\>|\mu\>,|i\>|\nu\>:i\in\{0,1\}^n\}$.
    \end{proof}

    Now we can prove Lemma \ref{lem:TestStateBasis}. Put $$\Psi(m,k) = \{|\psi_{m,x,b}(\mu,\nu)\>:\forall \mu<\nu\in\{0,1\}^k,\forall x\in\{0,1\}^n,\forall b\in\{0,1\}\}.$$ Obviously, for $m\neq m'$, the two sets $\Psi(m,k)$ and $\Psi(m',k)$ are disjoint to each other, i.e., $\Psi(m,k)\cap \Psi(m',k)=\emptyset$. So it is sufficient to prove that $\Psi(0,k)$ can be decomposed into the union of disjoint orthogonal bases:
    \begin{equation}\label{eq:strongBasis0k}
        \Psi(0,k) = C_1(0,k) \cup\cdots\cup C_{2^k-1}(0,k).
    \end{equation}
    Before continuing the proof, we define two concepts:
    \begin{itemize}
      \item a set $P=\{\lambda=\{\mu,\nu\}:\mu<\nu\in\{0,1\}^k\}$ is called a \textit{partition} of the set $\Lambda(k) = \{0,1\}^k$ if if for any $\lambda,\lambda'\in P$, $\lambda\cap \lambda'=\emptyset$ and $\bigcup_{\lambda\in P} \lambda= \Lambda(k).$
                 \item the set  generated by a couple $\lambda=\{\mu,\nu\}\in P$ is
          $$B(\lambda) = B_0(\mu,\nu)=\{|\psi_{0,x,b}(\mu,\nu)\>:\forall x\in\{0,1\}^n,\forall b\in\{0,1\}\}.$$
    \end{itemize}
    Since $\Lambda(k)$ has $2^k$ elements, there are always such partitions.

    We first prove that each partition corresponds to an orthogonal basis. By Lemma \ref{lem:weakBasis}, $C(\mu,\nu)$ is a orthogonal basis of the subspace spanned by $\{|\alpha\>|\mu\>,|\alpha\>|\nu\>: \forall |\alpha\>\in\hs_a\}$. This implies that the set $$B(P)=\bigcup_{\lambda\in P} B(\lambda)$$ is an orthogonal basis for any partition $P$.

    Observe that $\Psi(0,k)=\bigcup_{\lambda\in Q(k)} B(\lambda)$, where $Q(k)=\{\{\mu,\nu\}:\forall \mu<\nu\in\{0,1\}^k\}$. Therefore, Eq. \eqref{eq:strongBasis0k} can be derived by
    \begin{equation}\label{eq:strongBasisPartition}
        Q(k) = P_1\cup\cdots\cup P_{2^k-1},
    \end{equation}
    where $P_i\cap P_j=\emptyset$ for any $i\neq j$.
Thus, we prove Eq. \eqref{eq:strongBasisPartition} by induction on $k$.

    (1) For $k=1$, it is obvious, as $Q(1) = \{\{0,1\}\}$, i.e., it only contains one couple/set $\{0,1\}$.

    (2) Suppose Eq. \eqref{eq:strongBasisPartition} holds for $k=l$, i.e., $Q(l) = P_1\cup\cdots\cup P_{2^l-1}$. Then for $k=l+1$, we construct $P'_{i}$ in the following three cases:
    \begin{enumerate}
      \item $P'_1,\cdots,P'_u$, where $u=2^l$. In this case, $P'_i$ can be constructed as follows:
        \begin{equation}\label{eq:strongBasisPartition1}
            P'_i = \{\{\mu\cdot 0,\nu\cdot 0\},\{\mu\cdot 1,\nu\cdot 1\}:\forall \{\mu,\nu\}\in P_i\},
        \end{equation}
        where ``$\cdot$'' is the concatenation operator. For instance, if $\mu=01001\in\{0,1\}^5$, then $\mu\cdot 0 = 010010\in\{0,1\}^6$.
      \item $P'_{u+1},\cdots,P'_{2u}$ can be constructed as follows:
        \begin{equation}\label{eq:strongBasisPartition2}
            P'_i = \{\{\mu\cdot 0,\nu\cdot 1\},\{\mu\cdot 1,\nu\cdot 0\}:\forall \{\mu,\nu\}\in P_i\}.
        \end{equation}
      \item $P'_{v}$ with $v=2^{l+1}$ is
        \begin{equation}\label{eq:strongBasisPartition3}
            P'_v = \{\{\mu\cdot 0,\mu\cdot 1\}:\forall \mu\in\{0,1\}^l\}.
        \end{equation}
    \end{enumerate}
    After constructing such $P'_i$'s, it remains to prove that they satisfies Eq. \eqref{eq:strongBasisPartition}. First, we show that each $P'_i$ is a partition.
    \begin{itemize}
      \item Case 1. $i\leq u$: Suppose $\lambda_1=\{\mu'_1=\mu_1\cdot a_1,\nu'_1=\nu_1\cdot a_1\}\neq \lambda_2=\{\mu'_2=\mu_2\cdot a_2,\nu'_2=\nu_2\cdot a_2\}\in P'_i$, where $\mu_1,\nu_1,\mu_2,\nu_2\in\{0,1\}^l$, $a_1,a_2\in\{0,1\}$, and $\lambda_1\cap \lambda_2\neq \emptyset$. Since the join of $\lambda_1$ and $\lambda_2$ is nonempty, we have $a_1=a_2$. As a consequence, it holds that $\{\mu_1,\nu_1\}\cap\{\mu_2,\nu_2\}\neq\emptyset$. But this is not true because $P_i$ is a partition. A contradiction! So for any $\lambda_1\neq \lambda_2\in P'_i$, $\lambda_1\cap \lambda_2 = \emptyset$. Furthermore, as $|P'_i|=2|P_i|$ and $P_i$ is a partition of $\Lambda(l)$, we see that $P'_i$ is a partition of $\Lambda(l+1)$.
      \item Case 2. $u<i\leq 2u$: Similar to Case 1.
      \item Case 3. $i=v$: Obvious, as $P'_v$ can be written as $\{\{0,1\},\{2,3\},\cdots,\{2^{l+1}-2,2^{l+1}-1\}\}$.
    \end{itemize}

    Secondly, we prove that each couple $\{\mu'=\mu\cdot a,\nu'=\nu\cdot b\}$ belongs to one and only one $P'_i$, where $\mu'<\nu'$, $\mu,\ \nu\in\{0,1\}^l$ and $a,b\in\{0,1\}$. In fact, the existence of index $i=\textrm{Index}(\mu',\nu')$ for such set $P'_i$ can be verified as follows:
    \begin{enumerate}
      \item If $a = b$, then $\mu<\nu$, and $\textrm{Index}(\mu',\nu')\leq u$. By Eq. \eqref{eq:strongBasisPartition1}, we have $\textrm{Index}(\mu',\nu')=\textrm{Index}(\mu,\nu)$.
      \item If $a\neq b$ and $\mu\neq\nu$, then $\mu<\nu$, and $u<\textrm{Index}(\mu',\nu')\leq 2u$.  By Eq. \eqref{eq:strongBasisPartition2}, we have $\textrm{Index}(\mu',\nu')=\textrm{Index}(\mu,\nu)+u$.
      \item If $a\neq b$ and $\mu=\nu$, then $a=0$, $b=1$ and $\textrm{Index}(\mu',\nu')=2u+1=v$.
    \end{enumerate}

    The above two facts indicates that Eq. \eqref{eq:strongBasisPartition} holds for $k=l+1$. Consequently, by induction we complete the proof.

\subsection{Proof of Lemma \ref{lem:TestStateDis}}
    By Bayes' theorem, we obtain:
    \begin{align*}
        \Pr(m,x,b,\mu,\nu|M_v) &= \frac{\Pr(M_v|m,x,b,\mu,\nu)\Pr(m,x,b,\mu,\nu)}{\Pr(M_v)}\\ &= \frac{q\Pr(M_v|m,x,b,\mu,\nu)}{\Pr(M_v)},
    \end{align*}
    where we denote: $$q = \Pr(m,x,b,\mu,\nu) = \frac{1}{n2^{n+1}2^k(2^k-1)}.$$
    Now we only need to compute $\Pr(M_v)$ and $\Pr(M_v|m,x,b)$.
    By Lemma \ref{lem:TestStateBasis}, we have
    \begin{align*}
        \Pr(M_v) =& \sum_{l,y,a,\mu,\nu}\Pr(M_v|l,y,a,\mu,\nu)\Pr(l,y,a,\mu,\nu) = \sum q\Pr(M_v|l,y,a,\mu,\nu)\\
        =& \sum q \|M_v|\psi_{l,y,a}(\mu,\nu)\>\|^2 = \sum q \tr(M_v^\dag M_v|\psi_{l,y,a}(\mu,\nu)\>\<\psi_{l,y,a}(\mu,\nu)|)\\
        =& \sum_B q\tr(M_v^\dag M_v) = n(2^k-1) q\tr(M_v^\dag M_v),
    \end{align*}
    where $B$ ranges over all possible bases in Lemma \ref{lem:TestStateBasis}.
    Therefore,
    \begin{align*}
        \sum_{|\psi_{l,y,a}(\mu',\nu')\>\in B_0}\Pr(l,y,a,\mu,\nu|M_v)& = \sum\frac{q\Pr(M_v|l,y,a,\mu,\nu)}{\Pr(M_v)}\\ &= \frac{q\tr(M_v^\dag M_v)}{\Pr(M_v)}= \frac{1}{n(2^k-1)},
    \end{align*}
    where $B_0$ is the basis in Lemma \ref{lem:TestStateBasis} that contains $|\psi_{m,x,b}(\mu,\nu)\>$.
    Especially,
    \begin{equation*}
        \Pr(m,x,b,\mu,\nu|M_v)\leq\sum_{|\psi_{l,y,a}(\mu',\nu')\>\in B_0}\Pr(l,y,a,\mu,\nu|M_v) =  \frac{1}{n(2^k-1)}.
    \end{equation*}

\subsection{Proof of Theorem \ref{thm:Recovery}}

We first decompose $|\psi_{m,x,b}(\mu,\nu)\>$ into some other test states.
\begin{lemma}\label{lem:TestStateDec}
    For any $m\neq l\in \{0,\cdots,n-1\}$, $x\in\{0,1\}^n$, $b\in\{0,1\}$, and $\mu<\nu\in\{0,1\}^k$, we have:
    \begin{align}
        |\psi_{m,x,b}(\mu,\nu)\> =\ &(-1)^{bx_m}\frac{1}{2}(|\psi_{l,x',0}(\mu,\nu)\>+(-1)^{x_m}|\psi_{l,x',1}(\mu,\nu)\>\nonumber\\
        &+(-1)^{b}|\psi_{l,x'',0}(\mu,\nu)\>-(-1)^{b+x_m}|\psi_{l,x'',1}(\mu,\nu)\>),
    \end{align}
    where $x'=x_0x_1x_2\cdots x_{m-1}x'_{m}x_{m+1}x_{m+2}\cdots x_{l-1}x'_{l}x_{l+1}x_{l+2}\cdots x_{n-1}$ with $x'_m=x_m\oplus x_{l}$, $x'_l=0$, and $x''=x_0x_1x_2\cdots x_{m-1}x''_{m}x_{m+1}x_{m+2}\cdots x_{l-1}x''_{l}x_{l+1}x_{l+2}\cdots x_{n-1}$ with $x''_m=x_m\oplus x_{l}\oplus 1$, $x''_l=1$.
\end{lemma}
\begin{proof}
      First, we observe:
    \begin{equation}\label{eq:0+}
    \begin{cases}
        |0\>|+\> = \frac{1}{2}(|+\>|0\>+|+\>|1\>+|-\>|0\>+|-\>|1\>)\\
        |1\>|+\> = \frac{1}{2}(|+\>|0\>+|+\>|1\>-|-\>|0\>-|-\>|1\>)
    \end{cases}.
    \end{equation}
    For any $x$, we can rewrite $|\psi_{0,x,0}(\mu,\nu)\> $ as
    \begin{align}
        |\psi_{0,x,0}(\mu,\nu)\> =& \frac{1}{\sqrt{2}}Z(x)(|0\>|+\>^{\otimes n-1}|\vec{0}\>+|1\>|+\>^{\otimes n-1}|\vec{1}\>)\\
        =& \frac{1}{\sqrt{2}}Z_1^{x_1}(|0\>|+\>|\varphi\>|\vec{0}\>+(-1)^{x_0}|1\>|+\>|\varphi\>|\vec{1}\>),\label{eq:psi0x0}
    \end{align}
    where $|\varphi\>=\otimes_{j=2}^{n-1}(Z^{x_{j}}|+\>)$, and $Z_j$ represents Pauli $Z$ gate on $j$-th address qubit. By Eq. \eqref{eq:0+}, it can be further decomposed into
    \begin{align*}
        |\psi_{0,x,0}(\mu,\nu)\> =\ & Z_1^{x_1}(\frac{1}{2\sqrt{2}}(|+\>|0\>+|+\>|1\>+|-\>|0\>+|-\>|1\>)|\varphi\>|\vec{0}\>\nonumber\\
        &+(-1)^{x_0}\frac{1}{2\sqrt{2}}(|+\>|0\>+|+\>|1\>-|-\>|0\>-|-\>|1\>)|\varphi\>|\vec{1}\>).
    \end{align*}
    We permmute it and get
    \begin{align}
        |\psi_{0,x,0}(\mu,\nu)\> =& \frac{1}{2\sqrt{2}}Z_1^{x_1}( |+\>|0\>|\varphi\>|\vec{0}\>+(-1)^{x_0}|+\>|1\>|\varphi\>|\vec{1}\>) \nonumber\\
        &+\frac{1}{2\sqrt{2}}Z_1^{x_1}(|+\>|1\>|\varphi\>|\vec{0}\>+(-1)^{x_0}|+\>|0\>|\varphi\>|\vec{1}\>) \nonumber\\
        &+\frac{1}{2\sqrt{2}}Z_1^{x_1}(|-\>|0\>|\varphi\>|\vec{0}\>-(-1)^{x_0}|-\>|1\>|\varphi\>|\vec{1}\>) \nonumber\\
        &+\frac{1}{2\sqrt{2}}Z_1^{x_1}(|-\>|1\>|\varphi\>|\vec{0}\>-(-1)^{x_0}|-\>|0\>|\varphi\>|\vec{1}\>).
    \end{align}
    Then apply $Z_1^{x_0}$ and get
    \begin{align}
        |\psi_{0,x,0}(\mu,\nu)\> =& \frac{1}{2\sqrt{2}}(|+\>|0\>|\varphi\>|\vec{0}\>+(-1)^{x_0+x_1}|+\>|1\>|\varphi\>|\vec{1}\>) \nonumber\\
        &+\frac{1}{2\sqrt{2}}(-1)^{x_0}((-1)^{x_0+x_1}|+\>|1\>|\varphi\>|\vec{0}\>+|+\>|0\>|\varphi\>|\vec{1}\>) \nonumber\\
        &+\frac{1}{2\sqrt{2}}(|-\>|0\>|\varphi\>|\vec{0}\>-(-1)^{x_0+x_1}|-\>|1\>|\varphi\>|\vec{1}\>) \nonumber\\
        &-\frac{1}{2\sqrt{2}}(-1)^{x_0}(-(-1)^{x_0+x_1}|-\>|1\>|\varphi\>|\vec{0}\>+|-\>|0\>|\varphi\>|\vec{1}\>).
    \end{align}
    Extracting $U_{SWAP(0,1)}$, we have:
    \begin{align}
        |\psi_{0,x,0}(\mu,\nu)\> =& \frac{1}{2}U_{SWAP(0,1)}(\frac{1}{\sqrt{2}}(|0\>|+\>|\varphi\>|\vec{0}\>+(-1)^{x_0+x_1}|1\>|+\>|\varphi\>|\vec{1}\>) \nonumber\\
        &+\frac{1}{\sqrt{2}}(-1)^{x_0}((-1)^{x_0+x_1}|1\>|+\>|\varphi\>|\vec{0}\>+|0\>|+\>|\varphi\>|\vec{1}\>) \nonumber\\
        &+\frac{1}{\sqrt{2}}(|0\>|-\>|\varphi\>|\vec{0}\>-(-1)^{x_0+x_1}|1\>|-\>|\varphi\>|\vec{1}\>) \nonumber\\
        &-\frac{1}{\sqrt{2}}(-1)^{x_0}(-(-1)^{x_0+x_1}|1\>|-\>|\varphi\>|\vec{0}\>+|0\>|-\>|\varphi\>|\vec{1}\>)).
    \end{align}
    By using
    \begin{equation}
        |\psi_{0,x,1}(\mu,\nu)\> = \frac{1}{\sqrt{2}}Z_1^{x_1}((-1)^{x_0}|1\>|+\>|\varphi\>|\vec{0}\>+|0\>|+\>|\varphi\>|\vec{1}\>),
    \end{equation}
    together with Eq. \eqref{eq:psi0x0}, we obtain:
    \begin{align}
        |\psi_{0,x,0}(\mu,\nu)\>
        =\ & \frac{1}{2}U_{SWAP(0,1)}(|\psi_{0,x',0}(\mu,\nu)\>+(-1)^{x_0}|\psi_{0,x',1}(\mu,\nu)\>\nonumber\\&+|\psi_{0,x'',0}(\mu,\nu)\>-(-1)^{x_0}|\psi_{0,x'',1}(\mu,\nu)\>)\nonumber\\
        =\ &\frac{1}{2}(|\psi_{1,x',0}(\mu,\nu)\>+(-1)^{x_0}|\psi_{1,x',1}(\mu,\nu)\>\nonumber\\&+|\psi_{1,x'',0}(\mu,\nu)\>-(-1)^{x_0}|\psi_{1,x'',1}(\mu,\nu)\>),
    \end{align}
    where $x'=x'_0x'_1x_2x_3\cdots x_{n-1}$ with $x'_0=x_0\oplus x_1$, $x'_1=0$, and $x''=x''_0x''_1x_2x_3\cdots x_{n-1}$ with $x''_0=x_0\oplus x_1\oplus 1$, $x''_1=1$.
Since $|\psi_{0,x,b}\> = (-1)^{bx_0}X_0^b|\psi_{0,x,0}\>$, we have:
    \begin{align}
        |\psi_{0,x,b}(\mu,\nu)\>
        =\ & (-1)^{bx_0}X_0^b\frac{1}{2}U_{SWAP(0,1)}\nonumber
        (|\psi_{0,x',0}(\mu,\nu)\>\\ &+(-1)^{x_0}|\psi_{0,x',1}(\mu,\nu)\>+|\psi_{0,x'',0}(\mu,\nu)\>-(-1)^{x_0}|\psi_{0,x'',1}(\mu,\nu)\>)\nonumber\\
        =\ &(-1)^{bx_0}\frac{1}{2}U_{SWAP(0,1)}X_1^b(|\psi_{0,x',0}(\mu,\nu)\>+(-1)^{x_0}|\psi_{0,x',1}(\mu,\nu)\>\nonumber\\        &+|\psi_{0,x'',0}(\mu,\nu)\>-(-1)^{x_0}|\psi_{0,x'',1}(\mu,\nu)\>)\nonumber\\
        =\ &(-1)^{bx_0}\frac{1}{2}U_{SWAP(0,1)}((-1)^{bx'_1}|\psi_{0,x',0}(\mu,\nu)\>+(-1)^{bx'_1}(-1)^{x_0}|\psi_{0,x',1}(\mu,\nu)\>\nonumber\\
        &+(-1)^{bx''_1}|\psi_{0,x'',0}(\mu,\nu)\>-(-1)^{bx''_1}(-1)^{x_0}|\psi_{0,x'',1}(\mu,\nu)\>)\nonumber\\
        =\ &(-1)^{bx_0}\frac{1}{2}U_{SWAP(0,1)}(|\psi_{0,x',0}(\mu,\nu)\>+(-1)^{x_0}|\psi_{0,x',1}(\mu,\nu)\>\nonumber\\
        &+(-1)^{b}|\psi_{0,x'',0}(\mu,\nu)\>-(-1)^{b+x_0}|\psi_{0,x'',1}(\mu,\nu)\>)\nonumber\\
        =\ &(-1)^{bx_0}\frac{1}{2}(|\psi_{1,x',0}(\mu,\nu)\>+(-1)^{x_0}|\psi_{1,x',1}(\mu,\nu)\>\nonumber\\
        &+(-1)^{b}|\psi_{1,x'',0}(\mu,\nu)\>-(-1)^{b+x_0}|\psi_{1,x'',1}(\mu,\nu)\>)\label{eq:psi0xbto1x*}
    \end{align}

    Now we briefly consider the general case where $|\psi_{l,x,b}(\mu,\nu)\>$ are used to decompose $|\psi_{m,x,b}(\mu,\nu)\>$. The only thing that we need to do is to replace the 1-st (resp. 0-th) qubit by the $l$-th (resp. $m$-th) qubit in the above proof.
\end{proof}

A different decomposition can be done when $m$ does not change.
\begin{lemma}\label{lem:TestStateDec2}
    For any $m\in \{0,\cdots,n-1\}$, $x\in\{0,1\}^n$, $b\in\{0,1\}$, and $\mu<\nu\in\{0,1\}^k$, $\omega\in\{0,1\}^k\setminus\{\mu,\nu\}$, we have:
    \begin{align}
        |\psi_{m,x,b}(\mu,\nu)\> =\ &\frac{1}{2}(|\psi_{m,x,b'}(\mu,\omega)\>+(-1)^b|\psi_{m,x',b'}(\mu,\omega)\>\nonumber\\
        &+|\psi_{m,x,b''}(\omega,\nu)\>-(-1)^b|\psi_{m,x',b''}(\omega,\nu)\>),
    \end{align}
    where
    \begin{itemize}
      \item $x'=x'_0x_1x_2\cdots x_{n-1}$ with $x'_0=x_0\oplus 1$,
      \item $b' = b\oplus b_1$ with $b_1=0$ if $\mu<\omega$,  and $b_1=1$ if $\mu>\omega$,
      \item $b'' = b\oplus b_2$ with $b_2 = 0$ if $\omega<\nu$, and $b_2 = 1$ if $\omega>\nu$,
      \item we denote $|\psi_{m,x,b'}(\mu,\omega)\>=|\psi_{m,x,b'}(\omega,\mu)\>$ if $\mu>\omega$. The same for $\nu$ and $\omega$.
    \end{itemize}
\end{lemma}
\begin{proof}
    First, we observe:
    \begin{align*}
        \frac{|0\>|\mu\>+|1\>|\nu\>}{\sqrt{2}}=&\frac{1}{2}(\frac{|0\>|\mu\>+|1\>|\omega\>}{\sqrt{2}}+\frac{|0\>|\mu\>-|1\>|\omega\>}{\sqrt{2}}\\
        &+\frac{|0\>|\omega\>+|1\>|\nu\>}{\sqrt{2}}-\frac{|0\>|\omega\>-|1\>|\nu\>}{\sqrt{2}}).
    \end{align*}
    This directly leads to
    \begin{align*}
        V(\mu,\nu)|+\>^{\otimes n}|\vec{0}\> =& \frac{1}{2}(X_0^{b_1}V(\mu,\omega)+Z_0X_0^{b_1}V(\mu,\omega)\\
        &+ X_0^{b_2}V(\omega,\nu)-Z_0X_0^{b_2}V(\omega,\nu))|+\>^{\otimes n}|\vec{0}\>,
    \end{align*}
    where $b_1=0$ if $\mu<\omega$, $b_1=1$ if $\mu>\omega$, and $b_2 = 0$ if $\omega<\nu$, $b_2 = 1$ if $\omega>\nu$. Therefore,
    \begin{align*}
        Z(x)X_0^bV&(\mu,\nu)|+\>^{\otimes n}|\vec{0}\> \\
        =& \frac{1}{2}(Z(x)X_0^{b\oplus b_1}V(\mu,\omega)|+\>^{\otimes n}|\vec{0}\>+Z(x)X_0^bZ_0X_0^{b_1}V(\mu,\omega)|+\>^{\otimes n}|\vec{0}\>\\
        &+ Z(x)X_0^{b\oplus b_2}V(\omega,\nu)|+\>^{\otimes n}|\vec{0}\>-Z(x)X_0^bZ_0X_0^{b_2}V(\omega,\nu)|+\>^{\otimes n}|\vec{0}\>)\\
        =& \frac{1}{2}(Z(x)X_0^{b\oplus b_1}V(\mu,\omega)|+\>^{\otimes n}|\vec{0}\>+(-1)^bZ(x)Z_0X_0^{b\oplus b_1}V(\mu,\omega)|+\>^{\otimes n}|\vec{0}\>\\
        &+ Z(x)X_0^{b\oplus b_2}V(\omega,\nu)|+\>^{\otimes n}|\vec{0}\>-(-1)^bZ(x)Z_0X_0^{b\oplus b_2}V(\omega,\nu)|+\>^{\otimes n}|\vec{0}\>)
    \end{align*}
    and we complete the proof.
\end{proof}

%
%

    Now we are ready to prove Theorem \ref{thm:Recovery}. The control qibits $|c\>$ here can be ignored, since Bob can read $c$ by measurements without changing it. Suppose: (i) the correct test state is $|\psi_{m,x,b}(\mu,\nu)\>$; (ii) measurement operator $M_v$ is observed; and (iii) Bob sends the state $|\psi_{m',x',b'}(\mu',\nu')\>$ to Alice. Write $\Pr(|\psi_{m',x',b'}(\mu',\nu')\>\text{~passes~}|M_v)$ for  the probability that Bob success to pass the test in this case. First, we have:
    \begin{align*}
        \Pr(&|\psi_{m',x',b'}(\mu',\nu')\>\text{~passes~}|M_v)\\
        = &\sum_{m,x,b,\mu,\nu} \Pr(|\psi_{m',x',b'}(\mu',\nu')\>\text{~passes~}|m,x,b,\mu,\nu, M_v)\Pr(m,x,b,\mu,\nu| M_v)\\
        = &\sum_{m,x,b,\mu,\nu} \Pr(|\psi_{m',x',b'}(\mu',\nu')\>\text{~passes~}|m,x,b,\mu,\nu)\Pr(m,x,b,\mu,\nu| M_v),
    \end{align*}
    where the last equality is because of the following fact:
    \begin{itemize}
      \item once the original test state $|\psi_{m,x,b}(\mu,\nu)\>$ is fixed, Bob's success probability is independent of the measurement results and only dependent on what state he sends.
    \end{itemize}

    Now we compute the probability $\Pr(|\psi_{m',x',b'}(\mu',\nu')\>\text{~passes~}|m,x,b,\mu,\nu)$ that the correct test state is $|\psi_{m,x,b}(\mu,\nu)\>$, and Bob sends the state $|\psi_{m',x',b'}(\mu',\nu')\>$ to Alice:

   Case 1. $|\psi_{m',x',b'}(\mu',\nu')\>$ and $|\psi_{m,x,b}(\mu,\nu)\>$ are in the same basis of Lemma \ref{lem:TestStateBasis}. Then $\Pr(|\psi_{m',x',b'}(\mu',\nu')\>\text{~passes~}|m,x,b,\mu,\nu)\leq 1$.

    Case 2. $|\psi_{m',x',b'}(\mu',\nu')\>$ and $|\psi_{m,x,b}(\mu,\nu)\>$ are in different bases of Lemma \ref{lem:TestStateBasis}, and $m=m'$. There are two subcases:

    Subcase 2.1. Bob sends  $|\psi_{m',x',b'}(\mu',\nu')\>$  for both two test states. Then the best situation for Bob is $x'=x$, $b=b'$, and $\mu=\mu'$ without any loss of generality. (Another best situation is $x'=x$, $b=b'$, and $\nu=\nu'$.)  Then after $U_{SWAP(0,m)}$, $Z(x)$, $X_0^b$, $V(\mu,\nu)$ and $W$, the state $|\psi_{m',x',b'}(\mu',\nu')\>$ becomes
    \begin{equation*}
        \frac{1}{\sqrt{2}}(|+\>|0\>^{\otimes n+k-1}+|-\>|0\>^{\otimes n-1}|\nu\oplus\nu'\>).
    \end{equation*}
    Since $\nu\neq\nu'$ (Otherwise, it becomes Case 1), we have four different measurement outcomes:
    \begin{enumerate}
      \item $00\cdots 0$ on address and data qubits.
      \item $100\cdots 0$ on address and data qubits.
      \item $00\cdots 0$ on address qubits, and $\nu\oplus \nu'\neq 0\cdots 0$ on data qubits.
      \item $10\cdots 0$ on address qubits, and $\nu\oplus \nu'\neq 0\cdots 0$ on data qubits.
    \end{enumerate}
    Each of the four have probability 0.25. Then the situation that the two states pass the test only happens when both of the outcomes in Case 1 are observed, or both of the outcomes in Case 2 are observed. The corresponding probability is $\frac{1}{8}$.

   Case 2.2. Bob sends $|\psi_{m',x',b'}(\mu',\nu')\>$ for only one test state. Then at Step \ref{line:TTerminating}, no matter what test state Bob sends for the other one, the probability is at most 0.25 by the analysis of Case 2.1.

    Case 3.  $|\psi_{m',x',b'}(\mu',\nu')\>$ and $|\psi_{m,x,b}(\mu,\nu)\>$ are in different bases of Lemma \ref{lem:TestStateBasis}, and $m\neq m'$. Then there are also two subcases. The analysis is similar to Case 2, and the probability is at most 0.25.

Now by Lemma \ref{lem:TestStateBasis} and Lemma \ref{lem:TestStateDis}, we have:
    \begin{align*}
        \Pr(&|\psi_{m',x',b'}(\mu',\nu')\>\text{~passes~}|M_v)\\
        = &\sum_{m,x,b,\mu,\nu} \Pr(|\psi_{m',x',b'}(\mu',\nu')\>\text{~passes~}|m,x,b,\mu,\nu, M_v)\Pr(m,x,b,\mu,\nu| M_v)\\
        = &\sum_{m,x,b,\mu,\nu} \Pr(|\psi_{m',x',b'}(\mu',\nu')\>\text{~passes~}|m,x,b,\mu,\nu)\Pr(m,x,b,\mu,\nu| M_v)\\
        = & \sum_{B:|\psi'\>\in B} \Pr(|\psi_{m',x',b'}(\mu',\nu')\>\text{~passes~}|m,x,b,\mu,\nu)\Pr(m,x,b,\mu,\nu| M_v)\\
        &+\sum_{B:|\psi'\>\not\in B} \Pr(|\psi_{m',x',b'}(\mu',\nu')\>\text{~passes~}|m,x,b,\mu,\nu)\Pr(m,x,b,\mu,\nu| M_v)\\
        \leq & 1\times \frac{1}{K}+\frac{1}{4}\times \frac{K-1}{K} =\frac{1}{4}+\frac{3}{4n(2^k-1)},
    \end{align*}
    where $K=n(2^k-1)$ is the number of bases in Lemma \ref{lem:TestStateBasis}, $B:|\psi'\>\in B$ represents Case 1, and $B:|\psi'\>\not\in B$ represents Cases 2 and 3.

\subsection{Proof of Lemma \ref{lem:MultiRoundAttackImpossibility}}
    Suppose Alice send a test state $|\psi\>=\frac{1}{\sqrt{N}}\sum |i\>|\mu_i\>$ to Bob, and after Bob's action, Alice receives $|\psi'\>$. If $\mu_i\not\subseteq d$, then the states are the same, and Alice can not detect it. If $\mu_i\subseteq d$, then $$ |\<\psi'|\psi\>|^2 =(1-\frac{2}{N})^2=1-O(\frac{1}{N})$$ and it can be detected with probability at most $O(\frac{1}{N})$.

    Even if in one run of the inner protocol, Alice employs such tests many times, she still cannot detect it with a considerable probability. Suppose Bob's success probability is at least $1-\frac{a}{N}$, and Alice employs tests totally $b\sqrt{N}$ times, where $a,b$ are constants. Since there are at most $\frac{\pi}{4}\sqrt{N}$ iterations in an entire protocol, Bob successes to cheat in a run of the whole protocol with probability $(1-\frac{a}{N})^{b\sqrt{N}}\approx e^{-\frac{ab}{\sqrt{N}}}\approx 1$.

\section{Further Methods to Protect Bob's Privacy}\label{Apd:BobFurtherMethod}
Due to the limit of space, protection of Bob's privacy was only very briefly discussed in Section \ref{Sec:PrivacyBob}. Here, we continue to consider this issue.
If Alice is dishonest, one simple way for her to recover $f$ is as follows:
\begin{enumerate}
  \item For each loop $i$, Alice always employs two test rounds. Among them, one replaces the original computational round.
  \item Alice always chooses $\nu=\vec{1}$. Then $f_i(\mu)=1$ if and only if $f_i(\mu)=f_i(\nu)$.
  \item For the loops corresponding to the same control qubit, Alice chooses one fixed $\mu$, and gets $f_i(\mu)$ for all $i$. Then she can determine $f(\mu)$ by choosing the majority among these $f_i(\mu)$.
\end{enumerate}
This simple method can recover $f$ but with possible errors because a confusing qubit is added (see Definition \ref{defi:BobStrategy}). Alice can use some other methods to recover $f$ without errors. For instance, she can first get $f(\mu)$ for the same $\mu$ from the loops corresponding to three different control qubits. Since there is only one confusing qubit, two of these three values of $f(\mu)$ must be correct. So, she can get the correct value $f(\mu)$ for some $\mu$. Once she find some $\xi$ with $f(\xi)=0$, Alice may distinguish $h$ from $f$ as $h(\xi)=1\neq0=f(\xi)$, and then recover the entire $f$.

In remainder of this subsection, we present some further methods to preserve Bob's privacy.

\subsection{Adding a Second Confusing Qubit}
If Bob adds a second confusing qubit in his strategy (Definition \ref{defi:BobStrategy}), possibly only one quarter of functions $f_i$ corresponding to some control qubit may be $f$. Thus, Alice cannot get correct $f(\mu)$ by choosing the majority. The following are some sequences of the 16 functions corresponding to a control qubit:
\begin{align*}
    & f, f, a, h, a, h, f, f, g, g, b, h, b, g, g, h\\
    & g, g, a, h, a, g, g, h, f, f, b, h, b, h, f, f\\
    & f, h, g, a, h, a, g, f, g, b, h, b, g, h, f, f\\
    & h, h, f, f, a, a, h, a, a, h, f, f, h, b, h, b\\
    & h, h, f, f, g, g, h, g, g, h, f, f, h, f, h, f\\
    &\vdots
\end{align*}
In these sequences,
\begin{itemize}
  \item there are no fixed locations for $f$, and $f$ can be anywhere;
  \item $f$ can be either the minority or the majority;
  \item Moreover, it is impossible for Alice to get $f(\mu)$ by counting the number of ones or zeros for the value of $f_i(\mu)$, if we set $g=1-f$. This is because the number of ones or zeros can be any value from 4 to 12 when $g=1-f$. Since the distribution is symmetric, Alice cannot recover $f(\mu)$ by voting.
\end{itemize}
Therefore, by updating his strategy and carefully choosing the noises, Bob can prevent Alice from disclosing $f(\mu)$ through voting.

\subsection{Preventing Alice from Cheating in Computational Rounds}
Consider the case where there is only one control qubit, and the state is $|+\>\sum_j \alpha_j |j\>|0\>$. THen there is only one function $f$ to be applied.
\begin{itemize}
  \item If $f=h$, which corresponds to the identity operator $I_{a,d}$, then after the iteration $U_D(y)$, controlled $I_{a,d}$, $U_D(y)$ and controlled $\bar{G}$, the state becomes $$(|0\>\<0|\otimes I_{a,d}+|1\>\<1|\otimes \bar{G})|+\>\sum_j \alpha_j |j\>|0\>.$$
  \item If $f=\bar{h}=1-h$, which corresponds to the identity operator $-I_{a,d}$, then the state becomes
    \begin{align*}
        |+\>\sum_j \alpha_j |j\>|0\>
        \ra\ & |+\>\sum_j \alpha_j |j\>|d_{j\oplus y}\>\\
        \ra\ & \frac{1}{\sqrt{2}}(|0\>\otimes(I_{a,d}\sum_j \alpha_j |j\>|d_{j\oplus y}\>)+|1\>\otimes(-I_{a,d}\sum_j \alpha_j |j\>|d_{j\oplus y}\>))\\
        =\ &  |-\>\sum_j \alpha_j |j\>|d_{j\oplus y}\>
        \ra  |-\>\sum_j \alpha_j |j\>|0\>\\
        \ra\ & (|0\>\<0|\otimes I_{a,d}+|1\>\<1|\otimes \bar{G})|-\>\sum_j \alpha_j |j\>|0\>.
    \end{align*}
\end{itemize}
This fact means that Bob can control the final result by choosing his functions.
\begin{example}
    Suppose Bob wants to run Algorithm \ref{alg:ProtocolInner} with $T=8$ loops. Then there are three control qubits, denoted by $C_0$, $C_1$, $C_2$. Suppose Bob chooses his function as follows:
    \begin{itemize}
      \item $C_0$: $f'$, $h$, $f'$, $h$,
      \item $C_1$: $\bar{h}$, $h$,
      \item $C_2$: $\bar{h}$,
    \end{itemize}
    where $f'$ is an arbitrary function. After $8$ iterations at Step \ref{line:BMg}, the state becomes
    $$|+\>|-\>\otimes ((|0\>\<0|\otimes I_{a,d}+|1\>\<1|\otimes \bar{G})|-\>\sum_j \alpha_j |j\>|0\>|f(d_{j\oplus y})\>).$$
    Thus, if Bob undoes the operator $\bar{G}$ controlled by $C_2$, the control qubits becomes separable from the other part and the control state is $|+\>|-\>|-\>$.
\end{example}

In the above example, if Alice cheats in any computational round, the result becomes different. For instance, if Alice cheats on the computational round corresponding to $C_2$, then Bob's function $\bar{h}$ is applied on the cheating state but not the computational state. Thus, Alice has to guess what is the correct function to recover her cheating. Consequently, even if Alice knows that Bob randomly chooses $h$ (corresponding to result $|+\>$) or $\bar{h}$ (corresponding to result $|-\>$) in this step, it has probability 0.5 to be detected.
Therefore, Bob can use this method to detect Alice 's attacks with a high success probability.

\subsection{Restricting the Number of Alice's Tests}
After preventing Alice from cheating in the computational rounds, Bob can further reduce the chances that Alice can cheat. Note that there is at most one test round in each loop $i$, and this test round appears randomly with probability $2p$. Then by Chebyshev's inequality, we have:
\begin{equation*}
    \Pr(|n_t-2pT|\geq 6pT)\leq \frac{2p(1-2p)T}{(6pT)^2}=\frac{1-2p}{18pT},
\end{equation*}
where $n_t$ is the number of test rounds. In this paper, we usually set $p=0.05$, and $T$ should be $400/\sqrt{s_{\min}}$ as a second confusing qubit is added. So, if $s_{\min}=0.2$, $T$ should be 1024. By the above inequality, the probability that there are at least $0.4T$ test rounds in one run of Algorithm \ref{alg:ProtocolInner} is no more than $\frac{0.9}{0.95*1024}<0.001$.
Thus, Bob can count the number of test rounds in one run of Algorithm \ref{alg:ProtocolInner}. If it exceeds $0.4T \approx 410$, he may terminates the current run. The false positive probability is less than 0.001. This probability is tolerable, if Algorithm \ref{alg:ProtocolInner} is executed once. If the algorithm is executed $M$ times, this probability will be enlarged. For instance, if $M=100$, the total false positive probability may be approximately 0.095. Fortunately, it is still easy to deal with this false positive.
\begin{itemize}
  \item $M$ is small, say 100. For the first time the number of test rounds exceeds $0.4T$, Bob simply terminates the current run and start a new run. He announces that Alice is cheating by setting more test rounds when this situation happens twice. Then the total false positive probability is smaller than 0.005.
  \item $M$ is big, say $M>1000$. Bob can announce Alice's dishonesty when this excess happens $0.02M$. Then the total false positive probability is smaller than 0.0025.
\end{itemize}

\subsubsection{An Alternative Method}
Another way is to restrict the number of Alice's test rounds in a row. For instance, if test rounds are employed in six sequential loops, Bob terminates the algorithm. The false positive probability here is less than $1-(1-0.1^6)^{1024}<0.0011$.

\subsection{Summary}
Thus Bob can considerably reduce his privacy leakage
\begin{itemize}
  \item by adding a second confusing qubit and carefully choosing noise functions (In this way, Bob can make it impossible for Alice to recover any $f(\mu)$ by voting).
  \item by adding tests and counting the number of Alice's test rounds (In this way, Bob can further reduce the amount of information that Alice can get).
\end{itemize}

\begin{remark}
    It is worth noting that the above methods were not included in Algorithm \ref{alg:ProtocolInner}. If Bob directly use these methods, he might be treated as a dishonest data user. So, in order to make these methods work, Algorithm \ref{alg:ProtocolInner} should be modified.
\end{remark}

\section{More Discussions}

\subsection{Alice's Strategy to Detect Attack in Example \ref{exam:MRAttack1}}\label{Apd:Attackx=d}
An question was left open in Example \ref{exam:MRAttack1}: how Alice can detect an attack? Since the attack there is not very serious, here we only give an example rather than a formal protocol to deal with it.

If in one test, Alice employs $|\psi_{m,x,b}(\mu,d)\>$ (suppose $\mu<d$) as the test state with probability $O(\frac{1}{2^k})$, then she will find $f(d)\neq f(\mu)$  after receiving the returned test state. Thus,  she can flip one 0 to 1 in $d$  and obtain $d'$. She further employs $|\psi_{m,x,b}(\mu,d')\>$ as another test state. Since $f(x) = \delta(x,d)$, the second test state leads to $f(d')= f(\mu)$, which is not a result of any function indicating an inclusion relation $\subseteq$. So, the attack is detected. But such a detection does not really work in practice since its probability is $O(\frac{1}{2^k})$ and extremely low, as $d$ should be chosen randomly. This detection strategy was not included in the original protocol, since Alice changes $d$ for the second test state.

\subsection{Bob's Privacy without Noise and Tests}\label{Apd:PrivacyBob}
Bob's privacy was considered in Section \ref{Sec:PrivacyBob} with the assumption that he can adds noise and employs tests to protect himself.
Here, we further analyse Bob's privacy in the case where he is not allowed to add noise or to use any test.

Case 1. Alice is honest: All the information she gets about Bob's function $f$ is whether $f(\mu)=f(\nu)$ in the test rounds. Since she is honest, she will not use this information to compute the detailed form of $f$. So, Bob's privacy is preserved.

Case 2. Alice is semi-honest: She will employ all the legally derived information of $f(\mu)?=f(\nu)$ to compute $f$. We have showed in Section \ref{Sec:PrivacyBob} that in the best case for Alice, $k$ couples of $(\mu,\nu)$ are sufficient to recover $f$. On the other hand, in Algorithm \ref{alg:ProtocolInner}, the expected number of test states is $4(T+1)p$.
So, if the following condition is satisfied:
\begin{equation}\label{eq:BobsPrivacyCondition}
    4(T+1)p<k,
\end{equation}
then Alice cannot determine $f$ with certainty. Note that inequality (\ref{eq:BobsPrivacyCondition}) is true for a big database.
In fact, in a big database, since $N$ and $k$ may be very large, $s_{\min}$ must not be too small. This is because small $s_{\min}$ may result in too many rules mined by Bob, and these rules have a low support and thus are not important in practice \cite{AgrawalIS1993}. This problem is serious for a big database. So, Eq. \eqref{eq:BobsPrivacyCondition} is satisfied very likely.
For example, if $k=80$, $s_{\min} = 0.2$ and $p=0.05$, we have that $T=256$ ($>100/\sqrt{s_{\min}}$)  and then $4(T+1)p=51.4<80=k$.

One thing worth to mention is that the above analysis is just ideal. The pair $(\mu,\nu)$ is generated randomly. So, such random $k$ pairs may provide repetitive information, and the right side of Eq. \eqref{eq:BobsPrivacyCondition} is then much larger in practice.

Case 3. Alice is dishonest: She may generate $(\mu,\nu)$ by her strategy without randomness. Even further, she can make each round as a test round. Therefore, Alice can get explicitly $f$ more likely, as $k$ rounds are sufficient.

\end{document}